\documentclass[12pt]{article}

\usepackage{microtype} 
\usepackage{xr}
\usepackage{amsfonts}
\usepackage{amsmath}
\usepackage{amssymb,amsthm}
\usepackage{mathabx}
\usepackage[round]{natbib}
\usepackage{graphicx}
\usepackage[T1]{fontenc}
\usepackage{lmodern}
\usepackage{titlesec}
\usepackage{rotating,multirow}
\usepackage{geometry}
\usepackage{color}
\usepackage{longtable,lscape}
\usepackage[normalem]{ulem}
\usepackage[onehalfspacing]{setspace}{\color{blue}}
\usepackage{bbm}
\usepackage{caption}
\usepackage{subcaption}
\usepackage{tikz}
\usepackage{pgfplots}
\pgfplotsset{compat=newest}
\pgfplotsset{plot coordinates/math parser=false}
\newlength\figureheight
\newlength\figurewidth
\usepackage{accents}
\usepackage{float}
\usepackage{multibib}

\usepackage[todonotes={textsize=tiny}, defaultcolor=red]{changes}

\usepackage[pdftitle={Process Intangible},
  pdfauthor={},
  pdftoolbar=true,
  colorlinks=true,
  linkcolor=blue,
  citecolor=blue,
  bookmarksopen=true,
  bookmarksnumbered=true]{hyperref}
\usepackage{cleveref}

\usepackage{pdflscape}
\usepackage{epstopdf}
\usepackage{bm}
\usepackage{enumitem}
\usepackage{threeparttable}
\usepackage{booktabs}
\usepackage{makecell}

\usepackage{tikz}
\usetikzlibrary{matrix,positioning}
\tikzset{bullet/.style={circle,fill,inner sep=2pt}}

\newtheorem{theorem}{Theorem}[section]

\newtheorem{corollary}[theorem]{Corollary}

\newtheorem{definition}[theorem]{Definition}

\newtheorem{lemma}[theorem]{Lemma}

\newtheorem{proposition}[theorem]{Proposition}

\numberwithin{equation}{section} 

\makeatother

\begin{document}

\begin{titlepage}

\title{{\bf Flexible Information Acquisition \\in the Kyle Model}\footnote{First draft March 2026. Emails: S. Viswanathan \href{mailto:viswanat@duke.edu}{(viswanat@duke.edu)} 
 and Hao Xing \href{mailto:haoxing@bu.edu}{(haoxing@bu.edu)}.
}}

\author{
\large{S. Viswanathan}
\\
       \small{\sc Duke University \& NBER}\\
   \and
       \large{Hao Xing}
      \\
       \small{\sc Boston University}
       }

\date{March 2026}

\thispagestyle{empty}
\maketitle

{\vspace{-1 cm}}

\begin{abstract}



We study an information acquisition problem in which an informed trader acquires costly information prior to trading in the Kyle equilibrium. The cost of information acquisition is represented by an entropy cost. Regardless of the prior distribution of the asset payoff, continuous signals are optimal. Moreover, any continuously distributed signal, together with an associated logit-type posterior distribution of the payoff, yields the same ex-ante value for the informed trader, the same distribution of posterior expected payoff, and the same unconditional distribution of the informed trader’s trading strategy. Consequently, a normally distributed signal can be adopted without loss of generality. We further show that when the information acquisition cost increases or the volatility of noise trades decreases, the variance of the posterior expected payoff declines, the profit potential from trading diminishes, meanwhile the posterior expected payoff increasingly resembles a normal distribution, and the information leakage cost from trading decreases.

\end{abstract}



\end{titlepage}

\newpage
\setcounter{page}{1}
\section{Introduction}

Information acquisition in the pioneering \cite{Kyle1985} model is typically restricted to a decision problem in which the informed trader acquires costly normally distributed signals about the asset payoff. Specifically, the informed trader begins with a normal prior on the asset payoff and pays a cost to obtain a normal posterior, after which the Kyle trading game is played; see, for example, \cite{kyle1984market}, \cite{admati1988theory}, \cite{BackPedersen1998},  \cite{vives2014possibility}, and \cite{khorramithereze2025magic}. 

Rather than restricting attention to normal signals about a normally distributed payoff, we study information acquisition in the Kyle model when the informed trader has an arbitrary prior distribution on the asset payoff --- continuous or discrete, and possibly non-normal --- and can acquire signals with arbitrary marginal and joint distributions with the payoff. We follow \cite{Sims1998, Sims2003} to model the information acquisition cost using a mutual information cost. The informed trader designs the optimal signal structure by balancing the expected profit from trading using the observed signal and the ex-ante information acquisition cost. For this flexible information acquisition problem in the Kyle model, we provide sharp results on the optimal signal structure and the welfare of the informed trader. 

Our analysis is motivated by a recent development on characterizing Kyle's equilibrium from an optimal transport point of view. In an important paper, \cite{Back1992} extends the continuous-time \cite{Kyle1985} model by considering non-normal informed trader information. Recently, using an optimal transport approach, \cite{Back_et_al_2021} provided a representation for the ex-ante expected profit of the informed trader in \cite{Back1992} extension of the Kyle model. We use this representation as a starting point for our analysis. This representation suggests a tradeoff between the profit potential from trading and the information leakage cost associated with trading. 

On the one hand, the informed trader desires a precise signal so that the posterior distribution of the asset payoff has a high variance, which indicates a greater informational advantage and a higher profit potential in the subsequent Kyle trading game. However, the informed trader must also consider the inferences made by the market makers in this trading game. More aggressive trading leads to more precise inferences by market makers, increasing the cost of information leakage for the informed trader. Consider the case where the asset payoff follows a normal distribution. Extreme payoffs in the tails have small probabilities, but they generate large trading profits for the informed trader. Overweighting tail scenarios in the signal structure relative to the normal distribution and trading aggressively in those scenarios increases the informed trader's expected profit. However, such aggressive trading may be detected by the market makers, who adjust the price accordingly and thereby reduce the informed trader's profit. In the representation of \cite{Back_et_al_2021}, the profit potential from trading is measured by the variance of the asset payoff, and the information leakage cost is represented by the Wasserstein distance between the payoff distribution and the normal distribution of noise trades. In the absence of information acquisition cost, the former effect always dominates --- the informed trader maximizes the variance of the asset payoff using her precise information about the payoff.  

We consider a model in which the informed trader acquires a costly signal about the asset payoff prior to trading. Three competing forces shape the informed trader's information acquisition decision. First, the informed trader desires a precise signal that generates a high variance of the posterior conditional expectation of the asset payoff after signal observation, thereby increasing the profit potential from trading. Second, market makers' inference imposes a cost on the informed trader's information strategy. This cost is captured by the Wasserstein distance between the distribution of the posterior expected payoff and the distribution of noise trades. Third, the informed trader incurs an entropy cost that increases with the informational content of the signal. These three forces interact in forming the informed trader's information acquisition decision.  

We first consider the case in which the asset payoff has a normal prior distribution. In the rational inattention literature, following \cite{Sims1998, Sims2003}, it is well known that a decision maker with an entropy cost of information and a normal prior acquires information that leads to a normal posterior. However, in our problem, we must also consider the impact from the Wasserstein distance. To this end, the \cite{Gelbrich1990} bound implies that the distribution that minimizes the Wasserstein distance to the normal noise trade distribution is also normal. Therefore, the normal signal is optimal for the flexible information acquisition decision given a normal prior. 

We show that the optimal signal precision decreases with the information acquisition cost parameter, but increases with the standard deviation of the noise trades. When the information acquisition becomes more costly, the informed trader chooses a less precise signal, which reduces the variance of the posterior expected payoff and the profit potential from trading. When the noise trades become more volatile, the informed trader can better hide her informed orders among noise trades, hence the concern about information leakage is reduced. Consequently, the informed trader increases the signal precision to take advantage of the market environment. 

We then consider general payoff distributions. We show that the informed trader, no matter whether her prior distribution on the payoff is discrete or continuous, always chooses a continuous signal, which generates a continuously-distributed posterior expected payoff. At first glance, this result seems surprising for a discretely distributed prior. In the discrete choice literature, the recommendation lemma (see e.g. \cite[Lemma 1]{MatejkaMcKay2015}) implies that it is suboptimal to acquire distinct costly signals which lead to the same action, hence the same utility for the decision maker. However, because the noise trades in the Kyle model have a continuous distribution, price depends continuously on the informed trader's expectation of the asset payoff after observing the signal. Therefore, there is an infinite number of potential actions for the informed trader, each corresponding to a conditional expected payoff given signal realizations. Moreover, the informed trader's value is convex in the partition of the signal space, incentivizing the informed trader to use a continuously distributed signal. As a result, even when the prior distribution of the asset payoff is a two-point distribution on $\{L,H\}$, the optimal posterior expected payoff has a continuous distribution on the closed interval $[L,H]$. 

Regardless of the prior distribution on the asset payoff, we show that any continuous signal, which implies to a higher conditional expected payoff from a higher observed signal, leads to the same ex-ante value for the informed trader. Moreover, the distribution of the posterior expected payoff is invariant of the signal choice. This, perhaps surprising, indeterminacy in the signal choice resembles the independence of trading strategy with respect to the payoff distribution in \cite{Back1992}. For any continuously distributed payoff with a strictly increasing cumulative distribution function (CDF), the informed agent trades toward a target which shares the same distribution as the noise trades.\footnote{\cite{bagnoli2001existence} show that with one informed agent, the distribution of the informed agent's order must be identical to the noise trader distribution to to have a linear equilibrium.} This trading target is constructed by composing the quantile function of the payoff distribution with the CDF of the noise trade. This composition generates the same normal distribution for any continuously distributed payoff. This is also the reason why the informed trader's optimal value and the distribution of her conditional expected payoff do not depend on the signal choice in our setting. This composition function is the unique monotone optimal transport map from the conditional expected payoff distribution to the noise trade distribution; see \cite[Section 4]{galichon2016optimal} for optimal transport between 1-dimensional distributions. 

With non-normal priors (either discrete or continuous), the interaction between the profit potential (variance of the conditional expected payoff), the information leakage cost measured by the Wasserstein distance, and the information acquisition cost determine the shape of the posterior expected payoff. Specifically, consider a prior distribution that exhibits heavier tails than a normal distribution. When the information acquisition cost parameter is small, the informed trader acquires a precise signal that generates a high variance of the conditional expected payoff with excess kurtosis compared to a normal distribution. However, as the information acquisition cost parameter increases, the variance of the conditional expected payoff decreases, and the information leakage cost becomes increasingly important, so that the conditional expected payoff increasingly resembles a normal distribution, with its tails becoming lighter and the kurtosis converging to 3 (the kurtosis for normal distributions). As the information acquisition cost parameter approaches infinity, the information content of the signal vanishes, and the conditional expected payoff converges to the prior expected payoff. 

The impact of the noise trade volatility is exactly the opposite. As noise trades become more volatile, the cost of information leakage decreases, and the informed trader acquires a more precise signal that
 generates a heavy-tailed distribution of the conditional expected payoff. When the payoff has a two-state symmetric prior, the ratio between the information acquisition cost parameter and the standard deviation of the noise trade becomes a sufficient statistic for the conditional expected payoff distribution. A general implication of these results is that higher information acquisition costs  or lower noise trade volatility lead to posterior distributions of the expected payoff that are more similar to normal distributions.

Our paper presents a new flexible approach to information acquisition in the continuous time trading model of \cite{Kyle1985} and the extensions by \cite{Back1992} and \cite{Back_et_al_2021}. In particular, we allow for general priors and signal structures, either of them can be discrete, continuous, and non-normal. We provide a complete characterization of the optimal information acquisition decision. We show that the posterior expected payoff has a continuous distribution, even for discretely distributed priors. Further, there is an indeterminacy in the signal choice, though the distributions of the conditional expected payoff and the informed trader's trading strategy are both determined. Therefore, without loss of generality, the informed trader can pick a normal signal. This rationalizes the normal signal choice in \cite{kyle1984market} and \cite{BackPedersen1998}. We show that the flexible entropy cost mediates between the two tensions --- the desire for high trading profits from a precise signal, which generates a high variance of the conditional expected payoff, and the desire to reduce the information leakage cost induced by market makers' inference. Such a trade-off is robust for different types of prior payoff distributions. 

This paper sits at the intersection of three different literatures. First, our paper contributes to the information acquisition decision in the \cite{Kyle1985} model. Two notable recent contributions are \cite{Banerjee2020, Banerjee2022} and \cite{khorramithereze2025magic}. \cite{Banerjee2020} consider a dynamic information acquisition problem where the informed trader can either acquire a smooth flow of information or a lumpy information. They show that the informed trader increases the signal precision when the noise trading volatility is high. They assume that the payoff has a normal prior, the trader can acquire a normal signal about the payoff, and the information acquisition cost is a convex function of the signal precision. In our static information acquisition setting, general distributions of payoff and signal are considered. When the information acquisition incurs a fixed cost, \cite{Banerjee2020} show the equilibrium may fail to exist. \cite{Banerjee2022} model the costly research activity of the informed trader, who controls the research intensity and the likelihood of research success. Conditional on successful research, the informed trader observes a normal signal about payoff.
\cite{khorramithereze2025magic} consider a convex cost of acquiring normal signals in the generalized Kyle model studied by \cite{lambert2018strategic}. They tie the information efficiency in the large market limit with the  marginal cost of signal precision at the prior. When the information structure is exogenous, \cite{BackBaruch2004} consider a Bernoulli-distributed payoff for the informed trader's signal. \cite{CaldenteyStacchetti2010} examine informed trading with a  random announcement time. \cite{CollinDufresneFos2016} study informed trading with stochastic volatility in noise trading.

In the context of noisy rational expectations model (\cite{GrossmanStiglitz1980} and \cite{Hellwig1980}), signals on non-Gaussian payoffs have been considered by \cite{BreonDrish2015}, \cite{Chabakauri2022}, and \cite{Detemple2020}.  However, this literature does not consider the optimal information acquisition for non-Gaussian payoffs. A notable exception is \cite{Han2021}, who study a dynamic information acquisition problem in the setting of \cite{Hellwig1980} and \cite{BreonDrish2015}. \cite{Han2021} highlights the dynamic interaction of learning and price informativeness. 

Second, our paper is related to the rational inattention literature; see \cite{MackowiakMatejkaWiederholt2023} for an extensive review. \cite{MatejkaMcKay2015} provide a foundation for the multinomial logit model in a discrete choice setting. Different from \cite{MatejkaMcKay2015}, we show that the optimal signal has a continuous distribution even if the prior distribution on the state is discrete, due to the continuously distributed noise trades in the subsequent trading game. Our apporach is similar to the posterior approach by \cite{CaplinDean2013}, \cite{CaplinDeanLeahy2019, CaplinDeanLeahy2022}. However, the informed trader's problem does not lie in the posterior separable class studied in this literature. In contrast to this literature, where the net utility of the decision maker depends only on the posterior distribution of states, our net utility depends on the entire marginal of the signal that is chosen. 

Finally, our paper connects to the entropy regularized optimal transport and the static Schrodinger bridge problem considered in physics and machine learning. Given the marginal distributions of payoff and signal, we solve the Lagrangian multipliers for Bayes plausibility via the Sinkhorn problem; see \cite{Sinkhorn1964} and \cite{RuschendorfThomsen1993} for problems with discrete and continuous marginals. \cite{PeyreCuturi2019} provide an extensive review on the computational aspect of the Sinkhorn problem and its connection with the entropy regularized optimal transport. Optimal transport has been applied to matching problems and other topics in economics (see \cite{galichon2016optimal}).
The entropy regularized optimal transport has also been used to study the the Bayesian persuasion problem recently by \cite{justiniano2025entropy}.

The rest of the paper is organized as follows. Section \ref{sec:model} introduces the model setup and an equivalent formulation of the problem. Section \ref{sec:normal} considers the normally-distributed payoff. We start by assuming a normally-distributed signal and investigate the optimal signal precision. Then we prove that the normally-distributed signal is optimal for a normal prior. Section \ref{sec:discrete} presents the results for discretely-distributed payoffs. Section \ref{sec:cont_signal} presents the optimal signal structure among continuous signals. The optimality of continuous signals is presented in Section \ref{sec:cont_optimal}, followed by several examples in Section \ref{sec:cont_example}. Section \ref{sec:cont_payoff} studies the continuously distributed payoff. Section \ref{sec:conclusion} concludes.

\section{Model setup}\label{sec:model}
\subsection{Information acquisition and equilibrium}
The asset payoff is represented by a random variable $\tilde{v}$, whose value is realized and made public at the time $T$. At time $0$, the insider acquires a signal $\tilde{s}$. The state space for $\tilde{v}$ and $\tilde{s}$ is represented by the sets $\mathcal{V}$ and $\mathcal{S}$, respectively. We will specify these sets in different settings later. The marginal distribution of $\tilde{v}$ is denoted as $\{p(v)\}_{v\in \mathcal{V}}$, the marginal distribution of $\tilde{s}$ is $\{q(s)\}_{s\in \mathcal{S}}$, and the conditional distribution of $\tilde{v}$ given $\tilde{s}$ is $\{p(v|s)\}_{(v,s)\in \mathcal{V}\times \mathcal{S}}$.\footnote{When $\tilde{v}$ and $\tilde{s}$ have discrete distributions, $p(\cdot), q(\cdot)$, and $p(\cdot|\cdot)$ are probability distribution functions. When $\tilde{v}$ and $\tilde{s}$ have (conditional) continuous distributions, $p$ and $q$ represent the densities.}   

Information acquisition is costly, we follow the framework introduced by \cite{Sims1998, Sims2003} to model the information acquisition cost by mutual information. Let $\mathbb{H}(\mu) := - \sum_s \mu(s) \log \mu(s)$ be the entropy for the distribution $\mu$. The cost of acquiring the signal $\tilde{s}$ is given by
\begin{equation}
    \label{def:info_cost}
    \lambda I (\tilde{s}, \tilde{v}) := \lambda \big[\mathbb{H}(p) - \sum_s q(s) \mathbb{H}(p(\cdot| s))\big],
\end{equation}
where $\lambda>0$ represents the marginal cost of information. The difference $\mathbb{H}(p) -\sum_s q(s) \mathbb{H}(p(\cdot|s))$ represents the reduction in entropy from the prior distribution of the payoff to its posterior. A more accurate signal leads to a lower posterior entropy, hence a larger reduction from the prior entropy. 

The interaction among the informed trader, market makers, and noise traders is the same as in \cite{Kyle1985} and \cite{Back1992}. The cumulative demand from the informed trader is $X_t$ and the cumulative demand from all noise traders is $Z_t := \sigma_Z B_t$, where $B_t$ is a standard one-dimensional Brownian motion and independent with $\tilde{v}$. The cumulative aggregate demand observed by all market makers is $Y_t:= X_t + Z_t$. Competition and risk neutrality force market makers set the price as 
\begin{equation}\label{def:price}
 P_t = \mathbb{E} [\tilde{v}\,|\, \mathcal{F}^Y_t],
\end{equation}
where $\{\mathcal{F}^Y_t\}_{t\geq 0}$ is the filtration generated by $Y$ and it represents the information available to market makers. In the equilibrium considered, $Y_t$ is a sufficient statistic for $P_t$, i.e., the equilibrium price is given by $P_t = H(t,Y_t)$ for some function $H$. 

The informed trader observes the market price $P_t$ and computes $Y_t$, then she is able to infer $Z_t$. Hence, we assume that the informed trader observes $Z_t$ directly. The information available to the informed trader is the signal realization at time $0$, drawn from the distribution $\big(p(\cdot | \cdot), q(\cdot) \big)$ that she chooses, and the history of $Z$. The informed trader's realized profit for a given trading strategy $\theta_t$ is given by\footnote{\cite{Back1992} shows that it is suboptimal for the informed trader to submit orders so that $X_t$ has jumps or a nonzero quadratic variation. Therefore, without loss of generality, we represent the trading rate of the informed trader by a trading rate $\theta_t$.} 
\begin{equation*}
(\tilde{v}-P_T) X_T + \int_0^T X_{t-} dP_t = \int_0^T (\tilde{v} - P_t) \theta_t dt. 
\end{equation*}
The informed trader chooses the distribution $(p(\cdot|\cdot), q(\cdot))$ and a trading strategy $\theta_t$, measurable with respect to $\sigma(\tilde{s}) \vee \mathcal{F}^Z_t$, to maximize her (unconditional) expected profit net the information acquisition cost:
\begin{equation}
    \label{def:insider_obj}
    \mathbb{E} \Big[\int_0^T (\tilde{v} - P_t) \theta_t dt \Big] - \lambda I (\tilde{s}, \tilde{v}).
\end{equation}

An equilibrium among market participants is defined as follows.
\begin{definition}\label{def:equilibrium}
A Nash equilibrium is a triplet $(p(\cdot |\cdot), q(\cdot)), \theta_t$, and $P_t$ such that 
\begin{enumerate}
 \item[(i)] Given the informed trader's strategy $(p(\cdot|\cdot), q(\cdot))$ and $\theta_t$, the market makers set the price as in \eqref{def:price};
 \item[(ii)] Given the market makers' price setting strategy $P_t$, the informed trader chooses a distribution $(p(\cdot |\cdot), q(\cdot))$ and a trading strategy $\theta_t$ to maximize the objective in \eqref{def:insider_obj}.
\end{enumerate}
\end{definition}
Compared to the notion of equilibrium in \cite{Kyle1985}, the informed trader's information acquisition strategy $(p(\cdot|\cdot), q(\cdot))$ is incorporated into \eqref{def:equilibrium} and it is the main focus of our paper.

\subsection{An equivalent formulation}
In this section, we will introduce an equivalent formulation of an equilibrium. We call this equivalent formulation \emph{equilibrium with the estimated payoff}. In such an equilibrium, the payoff is replaced by the informed trader's best estimate of the payoff.

To introduce this equivalent formulation, we start from the informed trader's objective function. Because the informed trader's strategy $\theta_t$ is measurable with respect to $\sigma(\tilde{s})$, the law of iterated conditional expectation yields
\begin{equation*}
 \mathbb{E} \Big[\int_0^T \big(\tilde{v} - \mathbb{E}[\tilde{v}|\tilde{s}] \big) \theta_t dt \, \Big|\, \tilde{s}\Big]=0.
\end{equation*}
The previous identity and the risk neutrality of the informed trader imply that, conditioning on the signal observed, the expected payoff from trading is 
\[
\mathbb{E}\Big[\int_0^T \big(\tilde{v} - P_t \big)\theta_t dt \,\Big|\, \tilde{s} = s \Big]=\mathbb{E}\Big[\int_0^T \big(\mathbb{E}[\tilde{v}| \tilde{s}] - P_t \big)\theta_t dt \,\Big|\, \tilde{s} \Big].
\]
After taking the unconditional expectations on both sides of the previous equality, the informed trader's objective in \eqref{def:insider_obj} can be equivalently stated as 
\begin{equation}
    \label{def:insider_obj_trans}
    \mathbb{E} \Big[\int_0^T \big(\mathbb{E}[\tilde{v}|\tilde{s}] - P_t\big) \theta_t dt \Big] - \lambda I (\tilde{s}, \tilde{v}).
\end{equation}
Equation \eqref{def:insider_obj_trans} implies that the expected profit of the informed trader is the same if the asset payoff $\tilde{v}$ is replaced by its conditional expectation $\mathbb{E}[\tilde{v}|\tilde{s}]$, the best estimate (minimizing the mean squared error) for the payoff after the informed trader observes the signal. 

Given the price $p$, consider the informed trader's problem of designing the optimal trading strategy to maximize the expected profit. 
Consider a pricing rule
\begin{equation}\label{pricing_rule}
 P_t = H(t, Y_t),
\end{equation}
where $H(t,y) := \mathbb{E}\big[H(T,Z_T) \,|\, Z_t = y\big]$ for a function $H(T,\cdot)$.
Among all strategies satisfying `no doubling strategies' condition, \cite[Lemma 2]{Back1992} shows that the optimal trading strategy, which maximizes the conditional expected profit $\mathbb{E}\big[\int_0^T \big(\mathbb{E}[\tilde{v}|\tilde{s}] -P_t \big)\theta_t dt \,|\, \tilde{s}\big]$, satisfies 
\begin{equation}\label{P_T_id}
    P_T = H(T,Y_T) = \mathbb{E}[\tilde{v}|\tilde{s}].
\end{equation}
At the end of the trading horizon, informed agent's trading strategy reveals all her private information so that the market price converges to the informed agent's best estimate of the asset payoff. Equation \eqref{P_T_id} implies that $\mathbb{E}[\tilde{v}|\tilde{s}]$ is measurable with respect to $\sigma(Y_T)$. Hence $Y_T$ should be weakly more informative than $\tilde{s}$. Meanwhile, $Y_T$ is measurable with respect to the information generated by the history of noise trades $Z$ and $\tilde{s}$. Given that $\tilde{v}$ and $Z$ are independent, predicting $\tilde{v}$ using $Y_T$ should be less precise than predicting $\tilde{v}$ using $\tilde{s}$. Therefore, $Y_T$ should be weakly less informative than $\tilde{s}$. Combining both directions of the argument above, we obtain the following result.

\begin{lemma}\label{lem:cond_exp}
Suppose that 
$\mathbb{E}[\tilde{v}|\tilde{s}]$ is measurable with respect to $\sigma(Y_T)$. Then 
\begin{equation}
    \label{eq:cond_exp}
    \mathbb{E}[\tilde{v}|\tilde{s}] = \mathbb{E}[\tilde{v}|\mathcal{F}^Y_T].
\end{equation}
\end{lemma}

Suppose that the total demand $Y$, in its own filtration, has the same distribution as $Z$, then \eqref{pricing_rule}, \eqref{P_T_id}, and \eqref{eq:cond_exp} combined implies
\begin{equation}\label{price_cond_exp_id}
    P_t = H(t, Y_t) = \mathbb{E}\big[H(T,Y_T) \,|\, \mathcal{F}^Y_t\big] = \mathbb{E} \Big[\mathbb{E}\big[\tilde{v}| \tilde{s}\big]\, \Big| \,\mathcal{F}^Y_t\Big] = \mathbb{E} \Big[\mathbb{E}\big[\tilde{v}| \mathcal{F}^Y_T\big]\, \Big| \,\mathcal{F}^Y_t\Big] = \mathbb{E}\big[\tilde{v}\,|\, \mathcal{F}^Y_t\big],
\end{equation}
where second to the last equality holds thanks to \eqref{eq:cond_exp}, and the last equality follows from the law of iterated conditional expectation. Therefore, the pricing rule \eqref{pricing_rule} also satisfies the equilibrium requirement \eqref{def:price}. Equation \eqref{price_cond_exp_id} indicates that the pricing rule is the same when the payoff $\tilde{v}$ is replaced by its conditional expectation $\mathbb{E}[\tilde{v}|\tilde{s}]$. Therefore, market makers learn from the total order flow what the informed trader has learned from her private signal.

These observations on the informed trader's expected profit and the pricing rule motivate us to introduce the following notion of equilibrium.
\begin{definition}\label{def:equilibrium_2}
 A Nash equilibrium with the estimated payoff is a triplet $(p(\cdot|\cdot), q(\cdot))$, $\theta_t$, and $P_t$ such that 
 \begin{enumerate}
 \item[(i)] Given the informed trader's strategy $(p(\cdot |\cdot), q(\cdot))$ and $\theta_t$, the market makers set the price as
 \begin{equation}\label{def:price_2}
 P_t = \mathbb{E}\big[\mathbb{E}[\tilde{v}|\tilde{s}] \,\big|\, \mathcal{F}^Y_t\big].
 \end{equation}
 \item[(ii)] Given the market makers' price setting strategy $P_t$, the informed trader chooses a distribution $(p(\cdot |\cdot), q(\cdot))$ and a trading strategy $\theta_t$ to maximize the objective in \eqref{def:insider_obj_trans}
 \end{enumerate}
\end{definition}

When the informed trader observes the payoff $\tilde{v}$ at time $0$, \cite{Back1992} establishes a Markovian equilibrium, where the total order $Y$ is a sufficient static for market makers, and the informed agent's trade is inconspicuous, which implies the total order $Y$ in its own filtration has the same distribution as the noise order $Z$. Considering among the Markovian inconspicuous equilibria, we obtain the following result.
\begin{proposition}\label{prop:corres}
 The Markovian inconspicuous equilibria in Definitions \ref{def:equilibrium} and \ref{def:equilibrium_2} are 1-to-1 correspondence.
\end{proposition}

This result allows us to focus on searching an equilibrium with the estimated payoff introduced in Definition \ref{def:equilibrium_2}. For a given signal structure $(p(\cdot|\cdot), q(\cdot))$, the distribution of $\mathbb{E}[\tilde{v}|\tilde{s}]$ is fixed. An equilibrium between the optimal trading strategy, which maximizes the expected profit 
\[
\mathbb{E}\Big[\int_0^T \big(\mathbb{E}[\tilde{v}|\tilde{s}] -P_t \big)\theta_t dt\Big],
\]
and a pricing rule, which satisfies \eqref{def:price_2}, is a standard equilibrium in Kyle's model, where the asset payoff is replaced by its conditional expectation. In the next section, we recall a characterization of such an equilibrium by \cite{Back_et_al_2021}.

\subsection{Characterization of the Kyle's equilibrium via optimal transport}
\cite{Back_et_al_2021} provides a characterization of the Kyle's equilibrium from an optimal transport point of view. In particular, they identify a representation for the informed trader's ex ante expected profit, which is instrumental in our analysis of the optimal signal structure. 

Given distribution of $\mathbb{E}[\tilde{v}|\tilde{s}]$, let us first recall the main results of \cite{Back_et_al_2021} for a risk neutral informed trader. To this end, we denote $F_s$ as the distribution function of $\mathbb{E}[\tilde{v}|\tilde{s}]$ and $G$ as the distribution function of $Z_T$, which has a normal distribution $N(0, \sigma_Z^2 T)$. The following result restates \cite[Theorems 3.1 and 3.2]{Back_et_al_2021} in a one-dimensional setting with the asset payoff therein replaced by $\mathbb{E}[\tilde{v}|\tilde{s}]$. We also assume that $F_s$ is absolutely continuous and is strictly increasing. We will show in later sections that this without loss of generality when the signal $\tilde{s}$ is chosen optimally.

\begin{proposition}[\cite{Back_et_al_2021}]\label{prop:back_et_al}
Let $H(T, y) = F_s^{-1}\circ G(y)$ and consider the pricing rule in \eqref{pricing_rule}. The trading strategy $dX_t = \theta_t dt$ where 
\begin{equation}\label{opt:trading}
 \theta_t = \frac{\tilde{\zeta} - Y_t}{T-t}, \quad \tilde{\zeta} = G\circ F_s^{-1} \big( \mathbb{E}[\tilde{v}|\tilde{s}]\big),
\end{equation}
maximizes the informed trader's expected profit $\mathbb{E}\big[\int_0^T \big(\mathbb{E}[\tilde{v}|\tilde{s}] -P_t \big)\theta_t dt \,|\, \tilde{s}\big]$ in the class of trading strategies satisfying $\mathbb{E}\big[\int_0^T H(t, X_t + Z_t) dZ_t \big]=0$. Given the trading strategy in \eqref{opt:trading}, $Y$, in its own filtration, has the same distribution as $Z$, and the pricing rule \eqref{pricing_rule} satisfies the equilibrium condition \eqref{def:price_2}. Moreover, the informed trader's optimal (unconditional) expected profit is 
\begin{equation}\label{exp_pro_wass}
 \max_{\theta} \mathbb{E}\Big[\int_0^T \big(\mathbb{E}[\tilde{v}|\tilde{s}] -P_t \big)\theta_t dt\Big] = \frac12 \Big\{\mathbb{E}\big[\mathbb{E}[\tilde{v}|\tilde{s}]^2\big] + \mathbb{E}[Z_T^2] - W_2^2(F_s, G) \Big\},
\end{equation}
where $W^2_2(F_s, G)$ is the squared Wasserstein-2 distance between $F_s$ and $G$. 
\end{proposition}

The construction of the equilibrium is obtained by \cite{Back1992}. \cite{Back_et_al_2021} provides an equilibrium characterization in a more general setting using the theory of optimal transport.\footnote{\cite{Back_et_al_2021} allow for multi-variate and discrete distributions for $\mathbb{E}[\tilde{v}|\tilde{s}]$. When $\mathbb{E}[\tilde{v}|\tilde{s}]$ has a discrete distribution, the optimal trading strategy still takes the form in \eqref{opt:trading}, but $\tilde{\zeta}$ is constructed using additional randomization.} The pricing rule $F_s^{-1}\circ G$ is the gradient of a convex function $\Gamma$, Which is call a Brenier potential. The pricing rule 
$\Gamma' = F_s^{-1} \circ G$ is the unique monotone map that transport $\tilde{\zeta}$ to $\mathbb{E}[\tilde{v}|\tilde{s}]$.\footnote{See \cite[Section 4]{galichon2016optimal} for more discussions on the 1-dimensional optimal transport.} The random variable $\tilde{\zeta}$ has the same distribution as $Z_T$. When the informed trader follows the strategy \eqref{opt:trading}, the total order $Y$ is a Brownian bridge satisfying $Y_T = \tilde{\zeta}$. Moreover, $Y$, in its own filtration, shares the same distribution as $Z$. Therefore, the equilibrium identified in Proposition \ref{prop:back_et_al} is a Markovian inconspicuous equilibrium.



Using the optimal transport theory, \cite{Back_et_al_2021} identifies a intriguing representation for the informed trader's optimal (unconditional) expected profit in \eqref{exp_pro_wass}. Let us first transform this representation using the centered expected payoff $\mathbb{E}[\tilde{v}|\tilde{s}] - \mathbb{E}[\tilde{v}]$. To this end, note that 
\[
\mathbb{E}\big[\mathbb{E}[\tilde{v}|\tilde{s}]^2\big] = \text{Var}\big(\mathbb{E}[\tilde{v}| \tilde{s}] \big) + \mathbb{E}\big[\mathbb{E}[\tilde{v}| \tilde{s}]\big]^2 = \text{Var}\big(\mathbb{E}[\tilde{v}| \tilde{s}] \big) + \mathbb{E}[\tilde{v}]^2.
\]
Meanwhile, the Mallows distance decomposition property for Wasserstein-2 distance implies that 
\[
W^2_2(F_s, G) = \mathbb{E}[\tilde{v}]^2 + W^2_2 (F_s^c, G)^2, 
\]
where $F_s^c$ is the distribution for the centered expected payoff $\mathbb{E}[\tilde{v} \, |\, \tilde{s}] - \mathbb{E}[\tilde{v}]$.
Plugging the previous two equations into the right-hand side of \eqref{exp_pro_wass}, the unconditional expected profit of the insider is transformed to
\begin{equation}\label{insider_exp_pro_cen}
 \frac12 \Big\{\text{Var}\big(\mathbb{E}[\tilde{v} | \tilde{s}]\big) + \mathbb{E}\big[Z_T^2\big] - W_2^2 (F_s^c, G) \Big\}.
\end{equation}

The decomposition \eqref{insider_exp_pro_cen} highlights the impact of the conditional expected payoff $\mathbb{E}[\tilde{v}|\tilde{s}]$ and the noise trade distribution $G$ on the informed trader's expected profit. The variance of the conditional expected payoff increases the expected profit of the informed trader. The larger the variance, the more informational advantage the informed trader has over the market makers. After observing her private signal, the informed trader updates her payoff estimate to $\mathbb{E}[\tilde{v}|\tilde{s}]$. Meanwhile, market markers only know the prior expectation $\mathbb{E}[\tilde{v}]$. A higher variance of $\mathbb{E}[\tilde{v}|\tilde{s}]$ indicates a higher information asymmetry, hence a better trading opportunity for the informed trader. We call this the \emph{profit potential} of information. However, a higher variance of $\mathbb{E}[\tilde{v}|\tilde{s}]$ also introduces a higher cost to the informed trader. If the informed trader submits order aggressively, it is harder for her to hide her informed orders among noise trades, particularly if the distribution of trades in the tails differs from the normal distribution; the market makers learn informed trader's private information more efficiently, reducing the informational advantage of the informed trader. We call this the information \emph{leakage cost}. This effect introduces the third term in \eqref{insider_exp_pro_cen}. When the variance of $\mathbb{E}[\tilde{v}|\tilde{s}]$ increases, the centered distribution $F_s^c$ is further away from the noise orders' normal distribution $G$. This leads to a higher Wasserstein-2 distance, hence reduces the informed trader's expected profit. 

When the information acquisition cost is present, we will show in the next section that the information acquisition cost interacts with both the profit potential of information and the information leakage cost. The optimal design of the signal will balance the profit potential, the leakage cost, and the information acquisition cost. 

To conclude this section, we plug \eqref{exp_pro_wass} into \eqref{def:insider_obj_trans} to obtain the following information structure design problem for the informed trader:
\begin{align}
&\max_{p(\cdot |\cdot), q(\cdot)} \frac12 \Big\{\mathbb{E}\big[\mathbb{E}[\tilde{v}|\tilde{s}]^2\big] + \mathbb{E}[Z_T^2] - W_2^2(F_s, G) \Big\} -\lambda I(\tilde{s}, \tilde{v}), \label{signal_str_pro}\\
& \text{subject to} \nonumber\\
& \int q(s) ds = 1, \label{con:q_int} \\
& \int p(v|s) dv =1, \quad \text{for any } s,\label{con:p_int}\\
& \int p(v|s) q(s) ds = p(v), \quad \text{for any } v. \label{con:Bayes}
\end{align}
As we discussed before, the first three terms in \eqref{signal_str_pro} represents the informed trader's optimal (unconditional) expected profit from trading when the informed trader samples from a given signal distribution $(p(\cdot|\cdot), q(\cdot))$. The last term represents the information acquisition cost. The first two constraints ensures $q(\cdot)$ and $p(\cdot|\cdot)$ are probability densities. The last constraint is the Bayes formula between the posterior, marginal, and the given prior $p(\cdot)$. This is the Bayes plausibility condition.

In the next three sections, we will examine the trade-offs among the profit potential, the information leakage, and the information acquisitions cost. We will present the optimal signal structure for different payoff distributions. 

\section{Normal-distributed payoff}\label{sec:normal}

In this section, we consider the case where the payoff $\tilde{v}$ follows a normal distribution. We first illustrate the optimal signal precision when the signal has a normal distribution. Then we show that the normal signal distribution is optimal.

Consider a normal-distributed payoff $\tilde{v}\sim N(\mu_v, \sigma^2_v)$. Assume that the informed trader observes a signal 
\begin{equation}\label{normal:signal}
\tilde{s} = \tilde{v} + \tilde{\epsilon},
\end{equation}
where $\tilde{\epsilon}$ is independent of $\tilde{v}$ and follows a normal distribution $N(0,\sigma^2_{\epsilon})$. Let us first derive the optimal signal precision $1/\sigma^2_\epsilon$, then prove that the signal structure in \eqref{normal:signal} is optimal. 

Given the signal structure in \eqref{normal:signal}, the posterior distribution of $\tilde{v}$, conditioning on $\tilde{s}$, is normal with the mean and variance
\begin{align*}
 \mathbb{E}[\tilde{v}| \tilde{s}] =  \mathbb{E}[\tilde{v}] + \xi \big(\tilde{s} - \mathbb{E}[\tilde{v}]\big) \quad \text{and} \quad
 \sigma^2_{v|s} =  \sigma^2_{v} (1-\xi),
\end{align*}
respectively, where $\xi := \frac{\sigma^2_v}{\sigma^2_{s}} \in[0,1]$ and $\sigma^2_s := \sigma^2_v + \sigma^2_{\epsilon}$. Therefore, the conditional expectation of payoff $\mathbb{E}[\tilde{v} | \tilde{s}]$ follows a normal distribution with mean $\mathbb{E}[\tilde{v}]$ and variance $\xi^2 \sigma^2_s = \xi \sigma^2_v$. 

The informed trader's objective function in \eqref{signal_str_pro} can be simplified when both the payoff and the signal are normal distributed. To this end, given two normal distributions $F_s \sim N(\mathbb{E}[\tilde{v}], \xi \sigma_v^2)$ and $G \sim N(0, \sigma_Z^2 T)$, their Wasserstein-2 distance is 
\[
W_2^2(F_s, G) = \mathbb{E}[\tilde{v}]^2 + \big(\sigma_Z \sqrt{T} - \sqrt{\xi} \sigma_v\big)^2.
\]
Therefore, the informed trader's optimal (unconditional) expected profit in \eqref{exp_pro_wass} is 
\begin{align*}
\frac12 \Big\{\mathbb{E}\big[\mathbb{E}[\tilde{v}|\tilde{s}]^2\big] + \mathbb{E}[Z_T^2] - W_2^2(F_s, G) \Big\} = & \frac12 \Big\{\xi \sigma_v^2 + \mathbb{E}[\tilde{v}]^2 + \sigma_Z^2 T - \mathbb{E}[\tilde{v}]^2 -  \big(\sigma_Z \sqrt{T} - \sqrt{\xi} \sigma_v\big)^2\Big\} \\
= & \sigma_v\sigma_Z \sqrt{T} \sqrt{\xi}.
\end{align*}
In the previous equation, the net between the profit potential $\mathbb{E}\big[\mathbb{E}[\tilde{v}|\tilde{s}]^2\big]$ and the information leakage cost $W_2^2(F_s, G)$ increases with $\sqrt{\xi}$. Hence, a more precise signal increases the expected profit from trading.

For the information acquisition cost, recall that the entropy for a normal distribution $N(\mu, \sigma^2)$ is $\frac12 \log (2\pi  \sigma^2) + \frac12$. Therefore, the information acquisition cost is  
\begin{equation*}
 \lambda I(\tilde{s}; \tilde{v}) = \frac{\lambda}{2} \log \big(\frac{\sigma_v^2}{\sigma^2_{v|s}}\big) = \frac{\lambda}{2} \log \big( \frac{1}{1-\xi}\big).
\end{equation*}
Plugging the previous two equations into \eqref{signal_str_pro}, we obtain the following optimization problem for the informed trader:
\begin{equation}\label{signal_str_pro_normal}
\max_{\xi \in [0,1]} \Big\{ \sigma_v \sigma_Z \sqrt{T} \sqrt{\xi} - \frac{\lambda}{2} \log \big(\frac{1}{1-\xi}\big)\Big\}.
\end{equation}
The first term in the objective function of \eqref{signal_str_pro_normal} represents the expected profit. More precise signal yields a higher $\xi$, which improves the expected profit from trading. The second term in \eqref{signal_str_pro_normal} represents the information acquisition cost, whereas a higher $\xi$ means a more precise signal hence a higher cost. Therefore, the optimal choice of $\xi$ balances the marginal benefit and the marginal cost of information. The optimal $\xi$ and the optimal signal precision are presented in the following result.
\begin{proposition}\label{prop:normal}
 Suppose that the payoff $\tilde{v}$ follows a normal distribution $N(\mu_v, \sigma^2_v)$ and the informed trader observes a signal of the form \eqref{normal:signal}. Then the optimal conditional precision of the signal is 
 \begin{equation}\label{normal_opt_precision}
 \frac{1}{\sigma_{\epsilon}^2} = \frac{1}{\sigma_v^2} \frac{\xi^*}{1-\xi^*},
 \end{equation}
 where $\xi^*$ is the unique maximizer of the problem \eqref{signal_str_pro_normal}. This optimal precision and $\xi^*$ both decrease with $\frac{\lambda}{\sigma_v \sigma_Z \sqrt{T}}$.
\end{proposition}

The previous result shows that there is a negative relation between the optimal signal precision and $\frac{\lambda}{\sigma_v \sigma_Z \sqrt{T}}$. When $\lambda$ increases, information acquisition is more costly, the informed trader acquires less precise signal as a result. When $\sigma_v$ increases, $\xi^*$ increases as well, leading to a larger variance of the conditional expected payoff $\mathbb{E}[\tilde{v}|\tilde{s}]$. Even though a higher variance of $\mathbb{E}[\tilde{v}|\tilde{s}]$ increases the information leakage cost, the improvement in profit potential $\mathbb{E}\big[\mathbb{E}[\tilde{v}|\tilde{s}]^2\big]$ is more, leading to a higher expected profit for the informed trader. When $\sigma_Z \sqrt{T}$ increases, noise trade is more volatile, the informed trader can better hide her orders, hence the concern on information leakage is reduced. Therefore, the informed trader increases the signal precision to take advantage of this market environment. 

To close this section, we show in the next result that for normal distributed payoff, the normal signal distribution $(p(\cdot|\cdot), q(\cdot))$ is optimal among the class of distributions satisfying some non-degeneracy condition. Moreover, this optimal signal structure design can be implemented by the standard signal-plus-noise form in \eqref{normal:signal}.

\begin{proposition}\label{prop:normal_opt}
Suppose that the payoff is normal-distributed $\tilde{v} \sim N(\mu_v, \sigma_v^2)$. Then an optimal $(p(\cdot|\cdot), q(\cdot))$ for the problem \eqref{signal_str_pro}, together with the constraints \eqref{con:q_int} - \eqref{con:Bayes}, is normal, among all distributions so that $\frac{d}{ds}\mathbb{E}[\tilde{v}|\tilde{s} = s] \neq 0$ for almost all $s$. Moreover, this optimal signal structure can be implemented by \eqref{normal:signal}.
\end{proposition}

The restriction $\frac{d}{ds}\mathbb{E}[\tilde{v}|\tilde{s} = s] \neq 0$ for almost all $s$ is a non-degeneracy condition. Without this restriction, there can be two signals $s_1$ and
$s_2$ such that $\mathbb{E}[\tilde{v}|\tilde{s}=s_1] = \mathbb{E}[\tilde{v}|\tilde{s}=s_2]$. Information acquisition is costly, then it is wasteful for the informed trader to design two distinct signals which lead to the same estimates on the payoff. 

Intuition for Proposition \ref{prop:normal_opt} is the following. The Wasserstein-2 distance between the distribution $F_s$ for $\mathbb{E}[\tilde{v}|\tilde{s}]$ and the normal distribution $G$ for noise trades admits the \cite{Gelbrich1990} lower bound:
\begin{equation}\label{Gelbrich1990}
W^2_2(F_s, G) \geq (\mu_G - \mathbb{E}[\tilde{v}])^2 + (\sigma_G - \sqrt{K})^2,
\end{equation}
where $\mu_G = 0$ is the mean of $G$, $\sigma_G = \sigma_Z \sqrt{T}$ and $\sqrt{K}$ are the standard deviations for $G$ and $F_s$, respectively. Moreover, the lower bound is attained when $F_s$ is a normal distribution. We replace $W_2^2(F_s, G)$ in \eqref{signal_str_pro} by the lower bound in \eqref{Gelbrich1990}, hence obtain a new objective function which is an upper bound of the objective in \eqref{signal_str_pro}. The entropy information acquisition cost yields an normal distribution $(p^*(\cdot|\cdot), q^*(\cdot))$ as a maximizer for this upper bound objective function. Because $(p^*(\cdot|\cdot), q^*(\cdot))$ is normal, the Gelbrich lower bound is attained. Hence, $(p^*(\cdot|\cdot), q^*(\cdot))$ is also a maximizer for the original objective function in \eqref{signal_str_pro}. This optimal distribution can be implemented by the signal \eqref{normal:signal}. Hence, this signal structure, which we assumed for Proposition \ref{prop:normal}, is optimal.

\section{Discretely distributed payoff}\label{sec:discrete}

In this section, we consider a discretely distributed payoff $\tilde{v}$ with a finite state space $\mathcal{V} = \{v_1, \dots, v_N\}$. Even though the payoff state space is discrete and finite, our analytical result shows that optimal signals have continuous distributions. Moreover, any continuous signal leads to the same ex ante optimal value for the informed trader and the same distribution of the conditional expected payoff. 

\subsection{Continuous signal}\label{sec:cont_signal}
We consider a signal $\tilde{s}$ with a continuous distribution which admits a density with respect to the Lebesgue measure on $\mathbb{R}$. We denote it cumulative distribution function (CDF) as $Q$ and its probability distribution function (density) as $q$. The support of $q$ can be a subset of $\mathbb{R}$ or $\mathbb{R}$ itself. Because the signal states can be reordered, it is without loss of generality to consider the class of distributions such that $\text{support}(q) \ni s \mapsto \mathbb{E}[\tilde{v}| s]$ is strictly increasing.  

Let us first formulate the optimization problem for the informed trader. To this end,  for a given signal structure and the associated distribution of $\mathbb{E}[\tilde{v}|\tilde{s}]$, recall that the unconditional expected profit of the insider is given by \eqref{exp_pro_wass}. Using the quantile functions $F_s^{-1}$ and $G^{-1}$, the Wasserstein 2 distance between $F_s$ and $G$ can be expressed as 
\[
W^2_2(F_s, G) = \int_0^1 \big[F_s^{-1}(u) - G^{-1}(u)\big]^2 du.
\]
Therefore, the expected profit of the informed trader is 
\begin{align}
    \frac12 \Big\{\mathbb{E} \big[\mathbb{E}[\tilde{v}| \tilde{s}]^2\big] + \mathbb{E}[Z_T^2] - W^2_2(F_s, G) \Big\} = &\frac12 \int_0^1  \Big\{ \big(F_s^{-1}(u)\big)^2 + \big(G^{-1}(u) \big)^2 - \big[F_s^{-1}(u) - G^{-1}(u)\big]^2     \Big\} du \nonumber\\
    = & \int_0^1 F^{-1}_s(u) G^{-1}(u)\, du. \label{expected_profit_cont}
\end{align}
The first equality above follows from the fact that $\mathbb{E} \big[\mathbb{E}[\tilde{v}| \tilde{s}]^2\big] = \int_0^1 (F^{-1}_s(u))^2 du$ and $\mathbb{E}[Z_T^2] = \int_0^1 (G^{-1}(u))^2 du$. Using the quantile formulation in \eqref{expected_profit_cont}, the following result presents the informed trader's optimization problem.

\begin{lemma}\label{lem:pro_cont}
 The informed trader's optimization problem is 
 \begin{align}
 &\max_{p(\cdot|\cdot), q(\cdot)} \underbrace{\int \mathbb{E}[\tilde{v}|s] G^{-1}\big(Q(s)\big) q(s) \, ds}_{\text{Unconditional Expected Profit}} - \underbrace{\lambda \Big\{\sum_{n=1}^N \int \log \big(p(v_n|s) \big) p(v_n|s) q(s) \, ds - \sum_{n=1}^N \log \big( p(v_n)\big) p(v_n) \Big\}}_{\text{Information Acquisition Cost}},\label{obj_cont}\\
 &\text{subject to} \nonumber\\
 & \sum_{n=1}^N p(v_n|s) = 1, \quad \text{for any } s,\label{con:p_cont}\\
 & \int p(v_n|s) q(s) \, ds = p(v_n), \quad \text{for any } 1\leq n\leq N. \label{con:Baye_cont}
 \end{align}
\end{lemma}
The first term in \eqref{obj_cont} represents the unconditional expected profit from trading and the second term is the information acquisition cost. The constraint \eqref{con:p_cont} ensures that $p(\cdot|\cdot)$ is a posterior probability and the constraint \eqref{con:Baye_cont} represents the Bayes plausibility. The following result presents the informed trader's optimal choice of signal structure among continuous-distributed signals and the associated optimal value.

\begin{proposition}
    \label{prop:opt_cont}
    The following statements hold.
\begin{enumerate}
\item[(i)] The optimal posterior probabilities are given by 
    \begin{equation}
        \label{opt:post_cont}
        p(v_n | s) = \frac{e^{\frac{1}{\lambda}\big[v_n G^{-1}(Q(s)) + \mu_n\big]}}{\sum_{n'=1}^N e^{\frac{1}{\lambda}\big[v_{n'} G^{-1}(Q(s)) + \mu_{n'}\big]}}, \quad s \in \text{support}(q), n\in \{1,\dots, N\}.
    \end{equation}
    In the previous equation, the vector $(\mu_1, \dots, \mu_N)$ exists and is unique (up to additive constants, i.e., $(\mu_1, \dots, \mu_N)$ is identical to $(\mu_1 + c, \dots, \mu_N + c)$ for any constant $c$) such that the Bayes plausibility \eqref{con:Baye_cont} is satisfied. 
    
\item[(ii)] Any continuous signal $\tilde{s}$, such that $s\mapsto \mathbb{E}[\tilde{v}|s]$ is strictly increasing, achieves the same optimal value
    \begin{equation}
        \label{opt_value:cont}
         \lambda \, \mathbb{E}\Big[\log \Big(\sum_{n=1}^N e^{(v_n \tilde{z} + \mu_n)/\lambda} \Big)\Big] + \lambda \sum_{n=1}^N \log\big(p(v_n)\big) p(v_n) - \sum_{n=1}^N \mu_n p(v_n),
    \end{equation}
    where $\tilde{z}$ is a random variable with normal distribution $N(0,\sigma_Z^2 T)$. 
\item[(iii)] The conditional expected payoff $\mathbb{E}[\tilde{v}|\tilde{s}]$ has the same distribution as 
\begin{equation}
    \label{exp_payoff:cont}
     \frac{\sum_{n=1}^N v_n e^{(v_n \tilde{z} + \mu_n)/\lambda}} {\sum_{n=1}^N e^{(v_{n} \tilde{z} + \mu_{n})/\lambda}}. 
\end{equation}
In particular, any continuous signal $\tilde{s}$ leads to the same expected payoff distribution.
\end{enumerate}

\end{proposition}

It is surprising that any continuous signal achieves the same optimal value. To understand this result, we start from the optimization problem \eqref{obj_cont}. Introducing Lagrangian multipliers $\nu(s)$ for the constraint \eqref{con:p_cont} and $(\mu_1, \dots, \mu_N)$ for the constraint \eqref{con:Baye_cont}, we obtain the optimal posteriors in \eqref{opt:post_cont} from the first-order condition of $p(v_n|s)$. We also obtain the following objective function for the informed trader:
    \begin{equation}\label{obj:q_cont}
    \begin{split}
        \underbrace{\lambda \int \log \Big(\sum_{n=1}^N e^{\frac{1}{\lambda}\big(v_n G^{-1}(Q(s)) + \mu_n\big)} \Big) q(s) \, ds}_{\text{Expected Profit } + \lambda \,\mathbb{E}[\mathbb{H}(\tilde{v}|\tilde{s})]} + \underbrace{\lambda \sum_{n=1}^N \log\big(p(v_n)\big) p(v_n)}_{-\lambda \,\mathbb{H}(\tilde{v})} - \underbrace{\sum_{n=1}^N \mu_n p(v_n)}_{\text{Prior Adjustment}}.
    \end{split}
    \end{equation}
The first term represents the expected profit from trading and the expected entropy of the posterior payoff, when the informed trader chooses the optimal posterior as in \eqref{opt:post_cont}. The second term in \eqref{obj:q_cont} is proportional to the entropy of the prior payoff distribution. This term does not depend on choice of the informed trader.  The last term in \eqref{obj:q_cont} represents an adjustment term associated with the Bayes plausibility \eqref{con:Baye_cont}. The Lagrangian multipliers $(\mu_1, \dots, \mu_N)$ are chosen so that \eqref{con:Baye_cont} is satisfied. These Lagrangian multipliers introduce this adjustment term. In a two-state symmetric case (i.e., $N=2$ and $p(v_1) =p(v_2) =0.5$), due to the symmetric of noise trade distribution, Proposition \ref{prop:two_state_sym} in the appendix shows that $\mu_1=\mu_2=0$, hence this adjustment term vanishes.

Let us focus on the first term in \eqref{obj:q_cont}. Recall that $Q(\tilde{s})$ has a uniform distribution on $[0,1]$, when $\tilde{s}$ has a continuous distribution. Hence $G^{-1}(Q(\tilde{s}))$ has a normal distribution $N(0, \sigma_Z^2 T)$. Introducing a random variable $\tilde{z}$ with such a normal distribution, we obtain the following equivalent form for the first term in \eqref{obj:q_cont}:
\begin{align}
    \lambda \,\int \log \Big(\sum_{n=1}^N e^{\frac{1}{\lambda}\big(v_n G^{-1}(Q(s)) + \mu_n\big)} \Big) q(s) \, ds = & \lambda\,\mathbb{E} \Big[\log \Big(\sum_{n=1}^N e^{\frac{1}{\lambda}\big(v_n G^{-1}(Q(\tilde{s})) + \mu_n\big)} \Big)\Big] \nonumber\\
    = & \lambda\, \mathbb{E} \Big[\log \Big(\sum_{n=1}^N e^{\frac{1}{\lambda}\big(v_n \tilde{z} + \mu_n\big)} \Big)\Big].\label{first_term_obj_cont}
\end{align}
Combining \eqref{obj:q_cont} and \eqref{first_term_obj_cont} yields the value in \eqref{opt_value:cont}. Meanwhile, we show in the proof of Proposition \eqref{prop:opt_cont} that $(\mu_1, \dots, \mu_N)$ exists, unique (up to an additive constant), and is the same for any choice of continuously distributed signal $\tilde{s}$. Therefore, the value \eqref{opt_value:cont} is the same for any continuous signal, it is the optimal value for the informed trader's optimization problem in Lemma \ref{lem:pro_cont}.  

For any continuous signal $\tilde{s}$ and the associated optimal posterior in \eqref{opt:post_cont}, the conditional expected payoff is given by \eqref{exp_payoff:cont}. The distribution of the conditional expected payoff is the same for any continuous signal choice. This follows again from the fact that $G^{-1}(Q(\tilde{s}))$ has the same distribution as $\tilde{z}$ for any continuously distributed $\tilde{s}$. Given that the distribution of the conditional expected payoff is independent with the signal choice, both the unconditional expected profit and the information acquisition cost in \eqref{obj_cont} are independent of the signal choice as well, when the informed trader chooses the optimal posterior \eqref{opt:post_cont}.

Similar to the rational inattention literature, the first two terms in the objective \eqref{obj:q_cont} can be rewritten as \emph{posterior separable} form:
\[
\int V^{\text{net}}(s) q(s) \, ds - \lambda \mathbb{H}(\tilde{v}),
\]
where $V^{\text{net}}(s) := \lambda \log \Big(\sum_{n=1}^N e^{\frac{1}{\lambda}(v_n G^{-1}(Q(s)) + \mu_n)} \big)$ is the net utility associate with signal state $s$; see \cite{CaplinDean2013, CaplinDeanLeahy2019, CaplinDeanLeahy2022}. Different from these papers, where the net utility only depends on the posterior distribution of states, our net utility depends on the entire marginal for the signal via the signal CDF $Q$ and the Lagrangian multipliers $(\mu_1, \cdots, \mu_N)$.

Even though any continuous signal leads to the same optimal value for the informed trader, the optimal posterior in \eqref{opt:post_cont} depends on the choice of signal. Any pair of $\big(p(\cdot|\cdot), q(\cdot) \big)$ leads to the same optimal value. Hence the informed trader can choose the normal distributed $\tilde{z}$ as the signal. The associated optimal posterior is then
\begin{equation}
    \label{opt:post_cont_z}
p(v_n | z) = \frac{e^{(v_n z + \mu_n)/\lambda}}{\sum_{n'=1}^N e^{(v_{n'} z + \mu_{n'})/\lambda}}, \quad z \in \mathbb{R}, n\in \{1,\dots, N\}.
\end{equation}
This optimal posterior has a logit form. The optimal posteriors are determined by three forces. First, the posterior probability $p(v_n|z)$ increases with $v_n z$. When the informed trader receives a positive/negative signal, higher/lower payoff states receive larger posterior probabilities. When $z$ increases, high payoff states are associated with higher weights. In the limit $z\rightarrow \infty$, the highest payoff state receives the posterior probability one and other payoff states have zero posterior probability, i.e., $\lim_{z\rightarrow \infty}p(v_{\text{max}}|z) =1$, for $v_{\text{max}} = \max\{v_1, \dots, v_N\}$, and $\lim_{z\rightarrow \infty}p(v_n|z) =0$ for $v_n \neq v_{\text{max}}$.

Second, the distinction in the posterior probabilities $p(v_n|z)$ among different $v_n$ and $z$ reduces as $\lambda$ increases. A higher information acquisition cost increases the cost of distinguishing among different signal and payoff states. Therefore, the informed trader decreases the differences between different $p(v_n|z)$, i.e., the signal is less informative about the payoff. In the limit as $\lambda \rightarrow \infty$, $(\mu_1, \dots, \mu_N)$ adjusts so that the posterior $p(v_n|z)$ converges to the prior $p(v_n)$ for any $n$ and $z$. Hence the conditional expectation $\mathbb{E}[\tilde{v}|z]$ converges to the prior expectation $\mathbb{E}[\tilde{v}]$.  

Third, the vector $(\mu_1, \dots, \mu_N)$ represents the Lagrangian multipliers for the Bayes plausibility \eqref{con:Baye_cont}. Plugging \eqref{opt:post_cont_z} into \eqref{con:Baye_cont}, we obtain 
\begin{equation}\label{Sinkhorn-comp}
\int \xi_n k_{n}(z) \zeta(z) \,dz = p(v_n), \quad \text{where } \xi_n := e^{\mu_n/\lambda}, \, k_{n}(z) := e^{v_n z/\lambda}, \, \text{and} \, \zeta(z) := \frac{q(z)}{\sum_{n'=1}^N \xi_{n'} k_{n'}(z)},
\end{equation}
for $1\leq n\leq N$. These $N$ identities are equivalent to the following matrix identity
\begin{equation}\label{Sinknorn-Knopp}
\int\text{diag}(\xi_1, \dots, \xi_N) \, k(z) \, \zeta(z) \, dz= p,
\end{equation}
where $\text{diag}(\xi_1, \dots, \xi_N)$ is the diagonal matrix with diagonal entries $(\xi_1, \dots, \xi_N)$, $k(z)$ is a $N$-dimensional vector $(k_1(z), \dots, k_N(z))$, and $p$ is the column vector $(p(v_1), \dots, p(v_N))$. The identity \eqref{Sinknorn-Knopp} states that the row integral of  $\text{diag}(\xi_1, \dots, \xi_N) \, k(z) \, \zeta(z)$ is $p$. Meanwhile, the definition of $\zeta(z)$ implies that the column sums of $\text{diag}(\xi_1, \dots, \xi_N) \, k(z) \, \zeta(z)$ is $q(z)$. Problem \eqref{Sinknorn-Knopp} is the (continuous) \emph{Sinkhorn} problem, which matches a discrete marginal $p$ and a continuous marginal $q$. For given $k$, marginals $p$ and $q$, there exists a positive vector $(\xi_1, \dots, \xi_N)$ satisfying \eqref{Sinknorn-Knopp}. Moreover, this vector is unique up to a multiplicative scaling factor, i.e., if $(\xi_1, \dots, \xi_N)$ solves \eqref{Sinknorn-Knopp}, then so is $(c\,\xi_1, \dots, c\,\xi_N)$ for any nonzero $c$; see \cite{RuschendorfThomsen1993}.\footnote{Suppose that the support of the marginal $q$ is on finite number of points $\{s_1, \dots, s_M\}$. The Sinkhorn problem is also called the matrix scaling problem. When $M=N$, the existence and uniqueness of the solution for the Sinkhorn problem were obtained by \cite{Sinkhorn1964} and \cite{SinkhornKnopp1967}. Convergence of an iterative algorithm to identify the solution was also obtained in these papers. The existence and uniqueness of the solution are presented by \cite[Proposition 4.3]{PeyreCuturi2019} when $M$ and $N$ can be different. The result was generalized by \cite{RuschendorfThomsen1993} for general marginals. The convergence of the iterative algorithm for continuous marginals was proved by \cite{Ruschendorf1995}. We present the iterative algorithm to solve \eqref{Sinkhorn-comp} in Appendix \ref{app:sinkhorn_pos}. The Sinkhorn problem also relates to the entropic regularization of optimal transport; see \cite[Section 4]{PeyreCuturi2019} for the connection and further references.} From the definition of $\xi$'s, we see that there exists a vector $(\mu_1, \dots, \mu_N)$ so that $p$, defined in \eqref{opt:post_cont_z}, satisfies \eqref{con:Baye_cont}. Moreover, $(\mu_1, \dots, \mu_N)$ is unique up to an additive constant. 

Given $k$ and $q$, there is a positive relationship between $\mu_n$ and $p(v_n)$.\footnote{A proof is provided in Appendix \ref{app:sinkhorn_pos}.} Therefore, when the marginal $p(v_n)$ increases, $\mu_n$ increases as well, leading to a higher posterior $p(v_n| z)$ for all signal realizations. 

In an rational inattention discrete choice problem, \cite{MatejkaMcKay2015} obtain a logit form of optimal state-dependent choice probabilities, which provides a foundation for the multinomial logit model. Due to the recommendation lemma \cite[Lemma 1]{MatejkaMcKay2015}, signals one-to-one correspond to actions in \cite{MatejkaMcKay2015}, hence their optimal choice probabilities are the conditional probabilities of signals or actions given states; this is not the case in our problem. Different from \cite{MatejkaMcKay2015}, our posterior probabilities in \eqref{opt:post_cont_z} represent the conditional probabilities of states given signals.

Choosing the normal-distributed signal $\tilde{z}$, we show in the proof of Proposition \ref{prop:opt_cont} that $\mathbb{R} \ni z \mapsto \mathbb{E}[\tilde{v}|z]$ is strictly increasing. Hence the signal has a monotone relation with the informed trader's best estimate of the payoff. In particular, when the payoff has two states, i.e., $N=2$, the conditional expected payoff $\mathbb{E}[\tilde{v}|\tilde{z}]$ admits the following explicit density.

\begin{corollary}\label{cor:two_state_density}
 When $N=2$, the density of the conditional expected payoff $\mathbb{E}[\tilde{v}|\tilde{z}]$ is 
 \begin{equation}
     \label{two_state_density}
     f_s(v) = \frac{\lambda}{\sigma_Z \sqrt{T} (v_2 - v)(v-v_1)} \phi \Big(\frac{\lambda}{\sigma_Z \sqrt{T} (v_1 - v_2)} \log \Big(\frac{v_2 - v}{v-v_1} \Big) - \frac{\mu_1 - \mu_2}{\sigma_Z \sqrt{T} (v_1 - v_2)} \Big), 
 \end{equation}
 for $v\in(v_1, v_2)$, where $\phi(\cdot)$ is the density of the standard normal distribution.
\end{corollary}
We will illustrate several comparative statics for the conditional expected payoff after discussing the optimality of the continuous signals.

\subsection{Continuous signal is optimal}\label{sec:cont_optimal}
In our previous section, signal $\tilde{s}$ is assumed to have a continuous distribution. However, 
\begin{center}
\textit{
Is continuous signal optimal for a discretely-distributed payoff?}
\end{center}
To answer this question, we first consider discretely distributed signals: The signal $\tilde{s}$ could take $M$ different values $\{s_1, \dots, s_M\}$, where the number of signal states $M$ is not necessarily equal to the number of payoff states $N$. The probability distribution function of the discretely distributed signal is denoted as $q(s_m)$ and the CDF is $Q_m := \sum_{i=1}^m q(s_i)$ for $1\leq m \leq M$. We still denote the CDF for the conditional expected payoff $\mathbb{E}[\tilde{v}|\tilde{s}]$ as $F_s$.

For discretely distributed signals, let us first formulate the informed trader's optimization problem. To this end, note that the quantile function of $\mathbb{E}[\tilde{v}|\tilde{s}]$ takes the following form:
\begin{equation}\label{F_s_inv_disc}
    F_s^{-1} (u) = \mathbb{E}[\tilde{v} | s_m], \quad \text{ when } Q_{m-1} < u \leq Q_m.
\end{equation}
Therefore, the insider's unconditional expected profit in \eqref{expected_profit_cont} is
\begin{equation}\label{insider_exp_pro_disc}
\begin{split}
 \int_0^1 F_s^{-1}(u) G^{-1}(u) \, du = \sum_{m=1}^M \mathbb{E}[\tilde{v} | s_m]  \int_{Q_{m-1}}^{Q_m}  G^{-1}(u) du = \sum_{m=1}^M \mathbb{E}[\tilde{v}| s_m] I_m q(s_m),
\end{split}
\end{equation}
where $I_m := \frac{1}{q(s_m)} \int_{Q_{m-1}}^{Q_m} G^{-1}(u) du$. Incorporate the information acquisition cost, 
the insider's optimization problem is 
\begin{align}
&\max_{p(\cdot|\cdot), q(\cdot)} \sum_{m=1}^M \mathbb{E}[\tilde{v}| s_m] I_m q(s_m) - \lambda \Big\{\sum_{m=1}^M \sum_{n=1}^N \log \big(p(v_n|s_m) \big) p(v_n| s_m) q(s_m) -\sum_{n=1}^N \log \big(p(v_n)\big) p(v_n)\Big\},\label{insider_prob_disc}\\
& \text{subject to }\nonumber\\
 &\sum_{n=1}^N p(v_n| s_m) =1, \quad \text{for } 1\leq m\leq M,\label{con:p_prob}\\
 & \sum_{m=1}^M p(v_n| s_m) q(s_m) = p(v_n), \quad \text{for } 1\leq n\leq N.\label{con:bayesian}
\end{align}

The problem \eqref{insider_prob_disc} is the discrete analogue of \eqref{obj_cont}. To see the connection between these two problems, let us fix a signal state $s_m$ and let the number of signals go to infinite. In the limit, the probability of each signal, in particular $q(s_m)$, approaches zero. We then expect that 
\begin{equation}\label{lim_IM}
\lim_{M\rightarrow \infty} I_m = \lim_{M\rightarrow \infty} \frac{1}{q(s_m)} \int_{Q_{m-1}}^{Q_m} G^{-1}(u) du = G^{-1}(Q(s_m)).
\end{equation}
Changing the summations over $M$ in \eqref{insider_prob_disc} to integrals, we obtain the objective \eqref{obj_cont} for the continuous signal case.

The following result summarizes the optimal signal structure when only $M$ signal states are allowed.
\begin{proposition}
    \label{prop:disc_gen}
    The optimal posteriors for the signal state $s_m$ with $q(s_m)>0$ are given by 
    \begin{equation}
        \label{opt:post_disc}
        p(v_n| s_m) = \frac{e^{(v_n I_m + \mu_n)/\lambda}}{\sum_{n'=1}^N e^{(v_{n'}I_m + \mu_{n'})/\lambda}}, \quad 1\leq n\leq N, 1\leq m\leq M.
    \end{equation}
    In the previous equation, the vector $(\mu_1, \dots, \mu_N)$ exists and is unique (up to an additive constant) such that the Bayes plausibility \eqref{con:bayesian} is satisfied for given $q(s_1), \dots, q(s_N)$. The optimal value is
    \begin{equation}\label{obj:q_disc}
    \begin{split}
     \max_{q(\cdot)} \Big\{  \lambda \sum_{m=1}^M \log \Big(\sum_{n=1}^N e^{(v_n I_m + \mu_n)/\lambda} \Big) q(s_m) + \lambda \sum_{n=1}^N \log\big(p(v_n)\big) p(v_n) - \sum_{n=1}^N \mu_n p(v_n) \Big\},
    \end{split}
    \end{equation}
    and there exists an optimal $q$.
\end{proposition}
The optimal posterior in \eqref{opt:post_disc} and the objective in \eqref{obj:q_disc} are the discrete analogue of \eqref{opt:post_cont} and \eqref{obj:q_cont} respectively. 

Come back to our question whether continous signals are optimal for discretely distributed payoffs.  The following result provides the answer.

\begin{proposition}
    \label{prop:cont_optimal}
    When the payoff has a discrete distribution, continuous signals are optimal. The optimal posterior, optimal value, and the conditional expected payoff are presented in Proposition \ref{prop:opt_cont}
\end{proposition}

To understand the logic behind this result, recall the optimal value \eqref{obj:q_disc}. the informed trader picks the signal marginal $q$ to maximize the first term. The signal marginal is not only used to sum over different log terms, it also impacts $I_m = \frac{1}{q(s_m)} \int_{Q_{m-1}}^{Q_m} G^{-1}(du)$. For a signal state $s_m$, the informed trader can replace it by two new signal states $\underline{s_m}$ and $\overline{s_m}$ such that their associated probabilities satisfy $\underline{q_m} + \overline{q_m} = q(s_m)$, meanwhile all other signal states remain the same. Importantly, for given $(\mu_1, \dots, \mu_N)$, the function $x\mapsto \log \big(\sum_{n=1}^N e^{(v_n x+\mu_n)/\lambda} \big)$ is strictly convex. Therefore, it is strictly better off by splitting one signal states into two, i.e.,
\[
 \log \Big(\sum_{n=1}^N e^{(v_n I_m +\mu_n) /\lambda} \Big) q(s_m) < \log \Big(\sum_{n=1}^N e^{(v_n \underline{I_m}+\mu_n)/\lambda} \Big) \underline{q_m} + \log \Big(\sum_{n=1}^N e^{(v_n \overline{I_m}+\mu_n)/\lambda} \Big) \overline{q_m},
\]
where $\underline{I_m} := \frac{1}{\underline{q_m}} \int_{Q_{m-1}}^{Q_{m-1} + \underline{q_m}} G^{-1}(u) du$ and $\overline{I_m}:= \frac{1}{\overline{q_m}} \int_{Q_m - \overline{q_m}}^{Q_m} G^{-1}(u) du$ satisfy 
\[
I_m = \underline{I_m} \frac{\underline{q_m}}{q_m} + \overline{I_m} \frac{\overline{q_m}}{q_m}.
\]
Therefore, for given $(\mu_1, \dots, \mu_N)$, splitting signals is welfare improving. 

Consider Lagrangian multipliers $(\mu_1, \dots, \mu_N)$ as dual variables, we obtain from the argument above that the dual value function improves as the number of signals increases. Finally, we use a duality argument to show that the primal value function of the prior $p$ also weakly improves  when the number of signals increases, leading to the optimality of continuously distributed signals. 

When there is a finite number of possible actions, a recommendation lemma \cite[Lemma 1]{MatejkaMcKay2015} shows that it is optimal to design a signal one-to-one mapping to each action. This is because information acquisition is costly, it is suboptimal to distinguish two signals which lead to the same action and payoff. Therefore, the number of signals is equal to the number of actions, and discrete signals are optimal. In our setting, the conditional expected payoff $\mathbb{E}[\tilde{v}|\tilde{s}]$ corresponds to actions. Each realization of $\mathbb{E}[\tilde{v}|\tilde{s}]$ corresponds to a different trading strategy in \eqref{opt:trading} and a different conditional profit.  This matches the distribution of noise trades, which is also continuous. Therefore, the number of possible actions is infinite, and number of signals is infinite as well.

\subsection{Examples}\label{sec:cont_example}

In this section, we present comparative statics for several examples of the conditional expected payoff distribution. We will illustrate how the trade-off between the profit potential and the information leakage cost interact with the information acquisition cost in determining the conditional expected payoff.

Let us start from the case with two payoff states and a symmetric prior, i.e., $N=2$ and $p(v_1) = p(v_2) =0.5$ with $v_1 < v_2$. Proposition \ref{prop:opt_cont} implies that it is without loss of generality to consider a normal signal $\tilde{z}\sim N(0,\sigma_Z^2 T)$. Meanwhile, Proposition \ref{prop:two_state_sym} in the appendix shows that $\mu_1 = \mu_2=0$. Therefore, \eqref{exp_payoff:cont} implies that as $\lambda$ converges to zero, the posterior expected payoff $\mathbb{E}[\tilde{v}|\tilde{z}]$ converges to the prior. Specifically, $\lim_{\lambda \rightarrow 0} \mathbb{P}\big(\mathbb{E}[\tilde{v}| \tilde{z}<0] = v_1 \big) = \lim_{\lambda \rightarrow 0} \mathbb{P}\big( \mathbb{E}[\tilde{v}|\tilde{z}>0] = v_2\big) =0.5$. When the information acquisition cost vanishes, the posterior expected payoff is the low payoff state $v_1$ when the informed trader receives a negative signal and the high payoff state $v_2$ when the informed trader observes a positive signal. When $\lambda$ converges to infinity, \eqref{exp_payoff:cont} implies that the posterior expected payoff converges to the prior expectation $(v_1+v_2)/2$. When the information acquisition cost is infinite, the signal is not informative, hence the posterior expectation does not provide additional information beyond the prior mean. 

For an intermediate value of $\lambda$, the density of the posterior expected payoff $\mathbb{E}[\tilde{v}|\tilde{z}]$ is given by \eqref{two_state_density} with $\mu_1 = \mu_2 = 0$. From this expression, we can see that $\frac{\lambda}{\sigma_Z \sqrt{T}}$ is a sufficient statistic for the distribution of the expected payoff. Therefore, the information acquisition cost parameter has exactly the opposite effect as the noise trade volatility. When the information acquisition becomes less costly ($\lambda$ decreases), it has the same effect on the posterior expected payoff as an increasing noise trade volatility. When the noise trade volatility increases, inference by market makers becomes less precise, the informed trader can better hide her informed trades among noise trades and the information leakage cost is reduced. Therefore, the informed trader can acquire a more precise signal to exploit the profit potential. The limit $\lambda\rightarrow 0$ is equivalent to the limit $\sigma_Z \sqrt{T}\rightarrow \infty$. In such a limit, the noise trade is infinitely volatile, the informed trader acquires the perfect signal to maximize the variance of the posterior expected payoff, which indicates the highest profit potential.

\begin{center}
{\bf \autoref{fig:two_state_sym_lam} Here}
\end{center}

Figure \ref{fig:two_state_sym_lam} presents one example with two payoff states and a symmetric prior, but different information acquisition cost parameter $\lambda$. Panel (A) presents the densities of $\mathbb{E}[\tilde{v}|\tilde{z}]$ for different $\lambda$. They are all distributed symmetrically around zero, which is the prior expected payoff. Therefore, the skewness of $\mathbb{E}[\tilde{v}|\tilde{z}]$ is identically zero for different $\lambda$, as the black solid line in Panel (D) shows. However, when $\lambda$ is small, information acquisition is cheap. The informed trader acquires a precise signal, the resulting distribution of $\mathbb{E}[\tilde{v}|\tilde{z}]$ has a large variance and is close to the prior distribution of $\tilde{v}$. When $\lambda = 1$, Panel (A) shows that the density of $\mathbb{E}[\tilde{v}|\tilde{s}]$ peaks at $-2$ and $2$, which are the value of $v_1$ and $v_2$, respectively. However, the distribution of $\mathbb{E}[\tilde{v}|\tilde{z}]$ is distinct from the noise trade distribution $N(0,\sigma_Z^2 T)$, generating a large Wasserstein-2 distance between the centered conditional payoff distribution $F_s^c$ and the normal distribution of noise trades. Therefore, even though the informed trader benefits from a large profit potential from a precise signal and trades aggressively, she also faces high information leakage cost,  because the market markers learn more from her aggressive trades. In the limit $\lambda\rightarrow 0$, the profit potential dominates the information leakage cost. The informed trader acquires a prefect signal. The variance of $\mathbb{E}[\tilde{v}|\tilde{z}]$ is the variance of the prior payoff distribution, which is the largest possible variance of the conditional expected payoff.  

As $\lambda$ increases, information acquisition becomes increasingly costly, hence signal is less informative about payoffs. As Panel (B) of Figure \ref{fig:two_state_sym_lam} shows, the posterior probability of payoffs becomes less sensitive in signal, as $\lambda$ increases. As a result, Panel (A) shows that the distribution of $\mathbb{E}[\tilde{v}|\tilde{z}]$ becomes more concentrated around the prior mean and its variance decreases. As $\lambda$ increases, Panel (C) shows that the Wasserstein-2 distance between $F^c_s$ and $G$ has a U-shape. The reason is the following. When $\lambda$ is close to zero, the distribution of $\mathbb{E}[\tilde{v}|\tilde{s}]$ is close to the prior distribution of $\tilde{v}$ and far away from the distribution of noise trade $Z_T$. As $\lambda$ increases, the distribution of $\mathbb{E}[\tilde{v}|\tilde{s}]$ increasingly resembles the distribution of $Z_T$, reducing their Wasserstein-2 distance. However, as $\lambda$ continues to increase, difference between the variance of $\mathbb{E}[\tilde{v}|\tilde{s}]$ and $Z_T$ increases, which becomes the dominating component in $W_2(F_s^c, G)$. When we exclude the difference of mean and standard deviation, i.e., the Gelbrich lower bound, the orange dash line in Panel (C) shows that $W_2^2(F^c_s, G) - \mathbb{E}[\tilde{v}]^2 - \big(\text{std}(\mathbb{E}[\tilde{v}|\tilde{s}]) - \sigma_Z \sqrt{T}\big)^2$ vanishes as $\lambda$ increases. This implies that $\mathbb{E}[\tilde{v}|\tilde{s}]$ increasingly resembles a norm distribution with mean $\mathbb{E}[\tilde{v}]$ and standard deviation $\text{std}(\mathbb{E}[\tilde{v}|\tilde{s}])$. In particular, the orange dashed line in Panel (D) shows that the Kurtosis of $\mathbb{E}[\tilde{v}|\tilde{z}]$ converges to 3 (the Kurtosis of normal distributions) as $\lambda$ increase. Therefore, the informed trader extracts less profit potential from less precise signal, meanwhile the information leakage also reduces due to less aggressive trades. In the limit $\lambda \rightarrow \infty$, the informed trader no longer acquires signal, the conditional expected payoff $\mathbb{E}[\tilde{v}|\tilde{z}]$ converges to the prior expectation $\mathbb{E}[\tilde{v}]$, the information advantage disappears and the informed trader's value converges to zero.

\begin{center}
{\bf \autoref{fig:two_state_sym_sig} Here}
\end{center}

In the two-state symmetric prior case, $\frac{\lambda}{\sigma_Z \sqrt{T}}$ is a sufficient statics for the distribution of the conditional expected payoff. Therefore, decreasing $\sigma_Z$ has an equivalent impact as increasing $\lambda$. Figure \ref{fig:two_state_sym_sig} presents an example with different $\sigma_Z$. As $\sigma_Z$ decreases, it becomes harder for the informed trader to hide her order among noise trades. Panel (A) shows that the conditional expected payoff becomes more concentrated around the prior mean. Both the Wasserstein-2 distance $W_2(F_s^c, G)$ and the version excluding Gelbrich lower bound vanish, as $\sigma_Z$ approaches zero. The orange dashed line in Panel (D) shows that the Kurtosis of the conditional expected payoff also converges to 3. Meanwhile, as $\sigma_Z$ decreases, Panel (B) shows that the posterior probability of states become less sensitive in signal, i.e., signal becomes less informative. Even if the informed trader has precise information on the payoff, she cannot trade on it aggressively due to severe information leakage in a thin market with low volume of noise trades. Hence, the informed trader optimally acquires low precision signal to save on the information acquisition cost. 

\begin{center}
{\bf \autoref{fig:conv_sym} Here}
\end{center}

Figure \ref{fig:conv_sym} compares discrete and continuous signals. As we have seen from Proposition \ref{prop:cont_optimal}, as the number of signal state $M$ increases, the optimal value improves. This is confirmed by Panel (A). Such improvement in the optimal value is an joint effect of changes in the variance of the conditional expected payoff, the squared Wasserstein-2 distance between $F^c_s$ and $G$, and the information acquisition cost. As $M$ increases, Panel (B) shows that the variance of $\mathbb{E}[\tilde{v}|\tilde{s}]$ reduces. This is because when $M$ increases, there are more signal states with the posterior mean $\mathbb{E}[\tilde{v}|\tilde{s}]$ close to the prior mean $\mathbb{E}[\tilde{v}]$. In contrast, if there were only two signal states $s_1$ and $s_2$, the informed trader would spread out $\mathbb{E}[\tilde{v}|s_1]$ and $\mathbb{E}[\tilde{v}| s_2]$, in order to benefit from their associated profit potential. Hence, the informed trader would not choose signals such that either $\mathbb{E}[\tilde{v}|s_1]$ or $\mathbb{E}[\tilde{v}|s_2]$ is close to the prior mean $\mathbb{E}[\tilde{v}]$.  As a result, the variance of $\mathbb{E}[\tilde{v}|\tilde{s}]$ is larger when $\tilde{s}$ has only two states. As $M$ increases, the distribution of $\mathbb{E}[\tilde{v}|\tilde{s}]$ increasingly resemble a continuous distribution, hence $W^2_2(F_s^c, G)$ reduced, as Panel (C) shows. Meanwhile, when the informed trader chooses the optimal posterior and the signal marginal, Panel (D) shows that the information acquisition cost decreases with $M$ when $M\geq 3$. When $M$ increases, the informed trader redesign the posterior and the signal marginal, achieving a marginally lower information acquisition cost, even though the number of signals states is more. 

The reduction in both costs, mainly from the reduction of the squared Wasserstein-2 distance, dominates the decrease in the variance of conditional expected payoff, leading to a higher insider's value as the number of signal states increases.

We present several comparative statics when the prior distribution is asymmetric in Appendix \ref{app:asy_pri}. All comparative statics discussed above hold for the asymmetric prior case. Two new features emerge: i) due to the asymmetric prior, the distribution of $\mathbb{E}[\tilde{v}|\tilde{s}]$ is asymmetric as well. Therefore, $\mu_1$ and $\mu_2$ are nonzero, and \eqref{two_state_density} shows that $\frac{\lambda}{\sigma_Z \sqrt{T}}$ is no longer a sufficient statics for the conditional expected payoff distribution; ii) the skewness of $\mathbb{E}[\tilde{v}|\tilde{s}]$ converges to zero (the skewness for normal distribution) as $\lambda \rightarrow \infty$ or $\sigma_Z \rightarrow 0$.

\section{Continuously distributed payoff}\label{sec:cont_payoff}

Having understand the optimal signal structure for discretely distributed payoffs, we consider continuously distributed payoffs in this section. We show that the main result for discretely distributed payoffs (Proposition \ref{prop:opt_cont}) has a natural extension in the case with continuously distributed payoffs. 

Consider a continuously distributed payoff $\tilde{v}$ with a density $p$. We assume that $\mathbb{E}[\tilde{v}^2]$ is finite so that the informed trader's value in \eqref{signal_str_pro} is finite.\footnote{Jensen's inequality implies that $\mathbb{E}\big[\mathbb{E}[\tilde{v}|\tilde{s}]^2\big] \leq \mathbb{E}[\tilde{v}^2]$. Both the Wasserstein-2 distance and the information acquisition cost are nonnegative. Hence the value of \eqref{signal_str_pro} is bounded from above.} We focus on continuously distributed signals such that $\text{support}(q) \ni s \mapsto \mathbb{E}[\tilde{v}|s]$ is strictly increasing.\footnote{Proposition \ref{prop:cont_v} below show that the convexity structure of the informed trader's objective function with respect to signals remains the same when the payoff is continuously distributed. Therefore, the same statements of Proposition \ref{prop:cont_optimal} hold for continuously distributed payoffs, indicating that continuous signals are optimal.} Extending Lemma \ref{lem:pro_cont} to the current case, the informed trader's optimization problem is 
\begin{align}
 &\max_{p(\cdot|\cdot), q(\cdot)} \int \mathbb{E}[\tilde{v}|s] G^{-1}\big(Q(s)\big) q(s) \, ds - \lambda \Big\{\iint \log \big(p(v|s) \big) p(v|s) q(s) \, dsdv - \int \log \big( p(v)\big) p(v) \,dv\Big\},\label{obj_cont_v}\\
 &\text{subject to} \nonumber\\
 & \int p(v|s)\, dv = 1, \quad \text{for any } s,\label{con:p_cont_v}\\
 & \int p(v|s) q(s) \, ds = p(v), \quad \text{for any } v. \label{con:Baye_cont_v}
 \end{align}
Compared with \eqref{obj_cont} and its constraints, only the summations over different payoff states are replaced by integrals above. 

The following result is a natural extension of Proposition \ref{prop:opt_cont}.
\begin{proposition}\label{prop:cont_v}
Suppose that $\mathbb{E}[\tilde{v}^2]<\infty$. The following statements hold.
\begin{enumerate}
\item[(i)] The optimal posterior probabilities are 
\begin{equation}
    \label{opt:post_vont_v}
    p(v | s) = \frac{e^{\frac{1}{\lambda}\big[v G^{-1}(Q(s)) + \mu(v)\big]}}{\int e^{\frac{1}{\lambda}\big[v' G^{-1}(Q(s)) + \mu(v')\big]} dv'}, \quad s \in \text{support}(q), v\in \text{support}(p).
\end{equation}
In the previous equation, the function $v \mapsto \mu(v)$ exists and is unique up to additive constants.
\item[(ii)] Any continuous signals $\tilde{s}$, such that $s\mapsto \mathbb{E}[\tilde{v}|s]$ is strictly increasing, achieves the same optimal value 
\begin{equation}
        \label{opt_value:cont_v}
         \lambda \, \mathbb{E}\Big[\log \Big(\int e^{(v \tilde{z} + \mu)/\lambda} dv\Big)\Big] + \lambda \int \log\big(p(v)\big) p(v)\, dv - \int \mu(v) p(v) \,dv,
    \end{equation}
    where $\tilde{z}$ is a random variable with normal distribution $N(0,\sigma_Z^2 T)$.
\item[(iii)] The conditional expected payoff $\mathbb{E}[\tilde{v}|\tilde{s}]$ has the same distribution as 
\begin{equation}\label{exp_payoff_cont_v}
 \frac{\int v e^{(v \tilde{z} + \mu(v))/\lambda} dv} {\int e^{(v \tilde{z} + \mu(v))/\lambda} dv}. 
\end{equation}
In particular, any continuous signal $\tilde{s}$ leads to the same expected payoff distribution.
\end{enumerate}
\end{proposition}

When there are infinite payoff states, the Lagrangian multiplier $\mu$ becomes a function of states. For a given signal marginal $q$ and payoff a prior density $p$, the Sinkhorn problem becomes to find $v\mapsto \mu(v)$ so that 
\[
\int \xi(v) k(v,z) \zeta(z) \, dz = p(v), \quad \text{ for any v},
\]
where $\xi(v) = e^{\mu(v)/\lambda}$, $k(v, z) = e^{vz/\lambda}$, and $\zeta(z) = \frac{q(z)}{\int \xi(v) k(v,z) dv}$. The existence and uniqueness (up to additive constants) are obtained by \cite{RuschendorfThomsen1993}. 

Similar to the discretely distributed payoff case, any continuous signal yields the same optimal value. Therefore, without loss of generality, the informed trader can choose a normal signal $\tilde{z}\sim N(0,\sigma_Z^2 T)$, regardless of the payoff distribution. Therefore, restricting to a normally-distributed signal, such as \cite{kyle1984market} and \cite{BackPedersen1998}, may not restrict the generality of their results.

Figures \ref{fig:cont_lam} and \ref{fig:cont_sig} present an example where the payoff follows a double exponential distribution with density $\rho(v) = e^{-v}$ for $v>0$ and $\rho(v) = e^{v}$ for $v<0$. The informed trader chooses the normal signal $\tilde{z}\sim N(0,\sigma_Z^2 T)$. The comparative statics of the expected payoff distribution with respect to $\lambda$ and $\sigma_Z$ are consistent with those for discretely distributed payoffs in the previous section. 

\begin{center}
{\bf \autoref{fig:cont_lam} Here}
\end{center}

Figure \ref{fig:cont_lam} shows that as $\lambda$ increases, the expected payoff distribution increasingly resembles a normal distribution, with the kurtosis decreases from $6$ (the kurtosis for a double-exponential distribution) to $3$ (the kurtosis for a normal distribution). As the information acquisition becomes more costly, optimal signal becomes less informative, profit potential decreases, and the informed trader focuses more on preventing information leakage, so that the expected payoff becomes more similar to a normal distribution, in order to protect the small informational advantage of the informed trader.

\begin{center}
{\bf \autoref{fig:cont_sig} Here}
\end{center}

Figure \ref{fig:cont_sig} illustrates the opposite impact of $\sigma_Z$. As $\sigma_Z$ increase, the expected payoff increasingly resembles the prior distribution, with the kurtosis increases from $3$ to $6$. As the noise trades become more volatile, it becomes easier for informed trades to hide her trades among noise trades. The information leakage cost decreases, the informed trader acquires a more informative signal to increase the profit potential. 

\section{Conclusion}\label{sec:conclusion}
We study a flexible information acquisition problem in \cite{Kyle1985} model. The prior distribution of the asset payoff can be general and the signal choice is flexible; both can have discrete or continuous, including non-normal, distributions. The information acquisition cost is modelled by mutual information, following \cite{Sims1998, Sims2003}. We provide a complete characterization of the optimal information acquisition decision. We show that continuous signals are optimal regardless of the prior distribution of the asset payoff. Moreover, any continuous signals, which lead to a higher posterior expected payoff after a higher signal realization, lead to the same ex-ante value for the informed trader and the same distribution for the posterior expected payoff. Therefore, restricting to a normal signal is without loss of generality. We show that the entropy information acquisition cost mediates between two tensions --- the desire for high trading profits from a precise signal, which generates a high variance of the posterior expected payoff, and the desire to reduce the information leakage cost induced by market makers' inference. 

Although much of the existing literature adopts mutual information to represent information acquisition costs, some of the behavioral implications of this specification are inconsistent with experimental evidence (see, e.g., \cite{Woodford:2012}, \cite{CaplinDean2013}, and \cite{Dean_Neligh:2020, Dean_Neligh:2023}). In response to these findings, several studies propose more flexible formulations of information acquisition costs, including \cite{CaplinDean2013}, \cite{CaplinDeanLeahy2022}, \cite{Hebert_Woodford:2021}, and \cite{Pomatto:2023}. It would be interesting to investigate implications of these more general information acquisition costs on the optimal signal structure in the Kyle model. 

We consider a static information acquisition problem in this paper. When the informed trader can acquire multiple pieces of information during the trading game, such as the setting of \cite{Banerjee2020, Banerjee2022} or \cite{Han2021}, the incentive to acquire information interacts with future information acquisition and trading opportunities. For uniformly posterior-separable cost functions, a dynamic discrete choice problem is studied by \cite{Miao_Xing:2024}. Integrating the forward-backward algorithm introduced therein with the dynamic flexible information acquisition problem in the Kyle model is another interesting future research topic.


\begin{figure}[ht]
\centering
\includegraphics[scale=0.8]{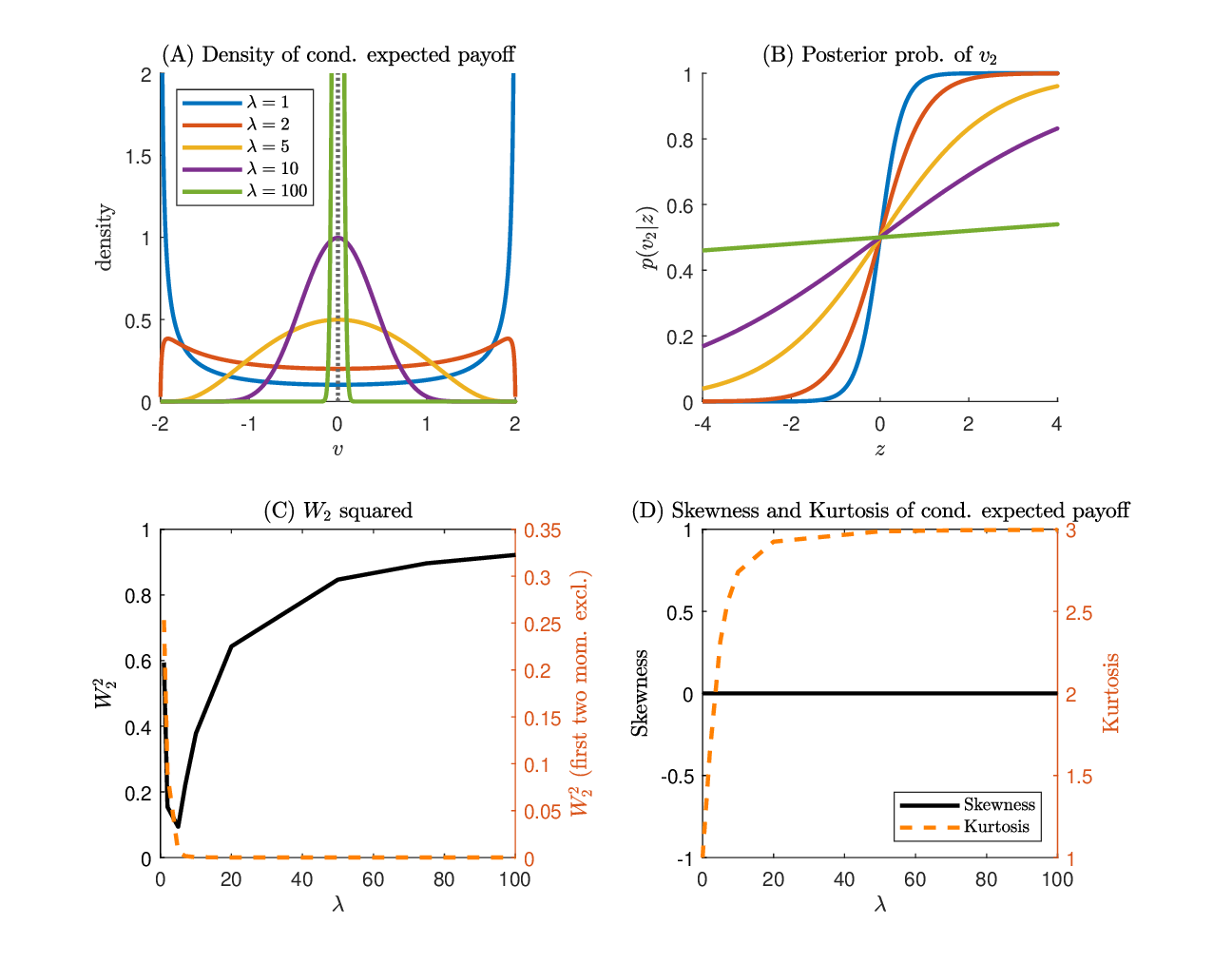}
\caption{Two-state symmetric prior (different $\lambda$)}
    \begin{minipage}{\textwidth}
	\footnotesize
	\vspace{2mm}
This figure presents the distribution of the conditional expected payoff, its skewness, Kurtosis, its $W_2$-distance with noise trade distribution, and the optimal posterior $p(v_2 | z)$ for the two-state symmetric prior case. The signal $\tilde{z}$ has a normal distribution $N(0,\sigma_Z^2 T)$.   
Parameters: $N=2$, $v_1 = -2, v_2 = 2$, $p(v_1) = p(v_2) = 0.5$, $\sigma_Z=1$, and $T=1$.
    \end{minipage}
    \label{fig:two_state_sym_lam}
\end{figure}

\begin{figure}[ht]
\centering
\includegraphics[scale=0.8]{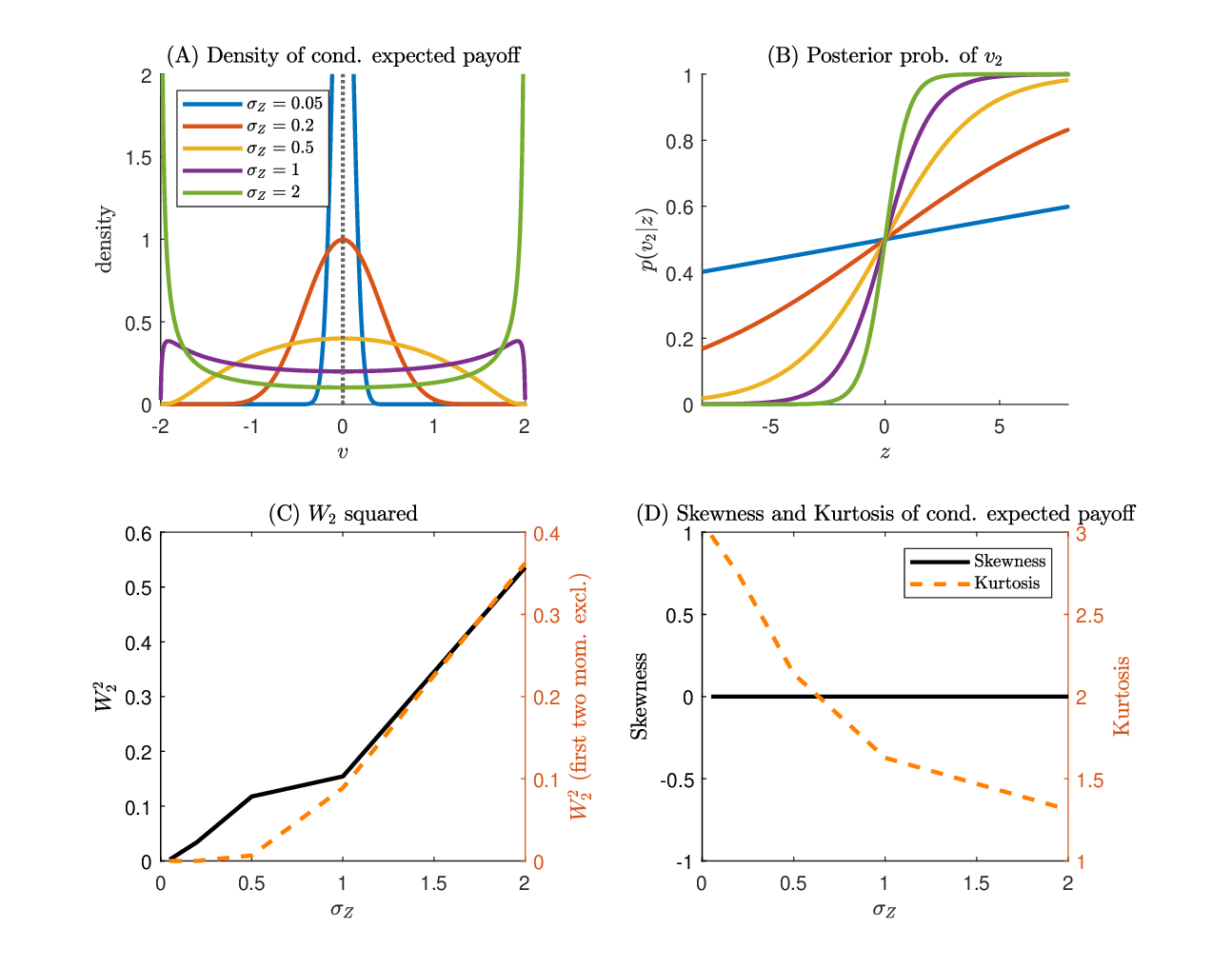}
\caption{Two-state symmetric prior (different $\sigma_Z$)}
    \begin{minipage}{\textwidth}
	\footnotesize
	\vspace{2mm}
This figure presents the distribution of the conditional expected payoff, its skewness, Kurtosis, its $W_2$-distance with noise trade distribution, and the optimal posterior $p(v_2 | z)$ for the two-state symmetric prior case. The signal $\tilde{z}$ has a normal distribution $N(0,\sigma_Z^2 T)$.   
Parameters: $N=2$, $v_1 = -2, v_2 = 2$, $p(v_1) = p(v_2) = 0.5$, $\lambda=2$, and $T=1$.
    \end{minipage}
    \label{fig:two_state_sym_sig}
\end{figure}

\begin{figure}[ht]
\centering
\includegraphics[scale=0.7]{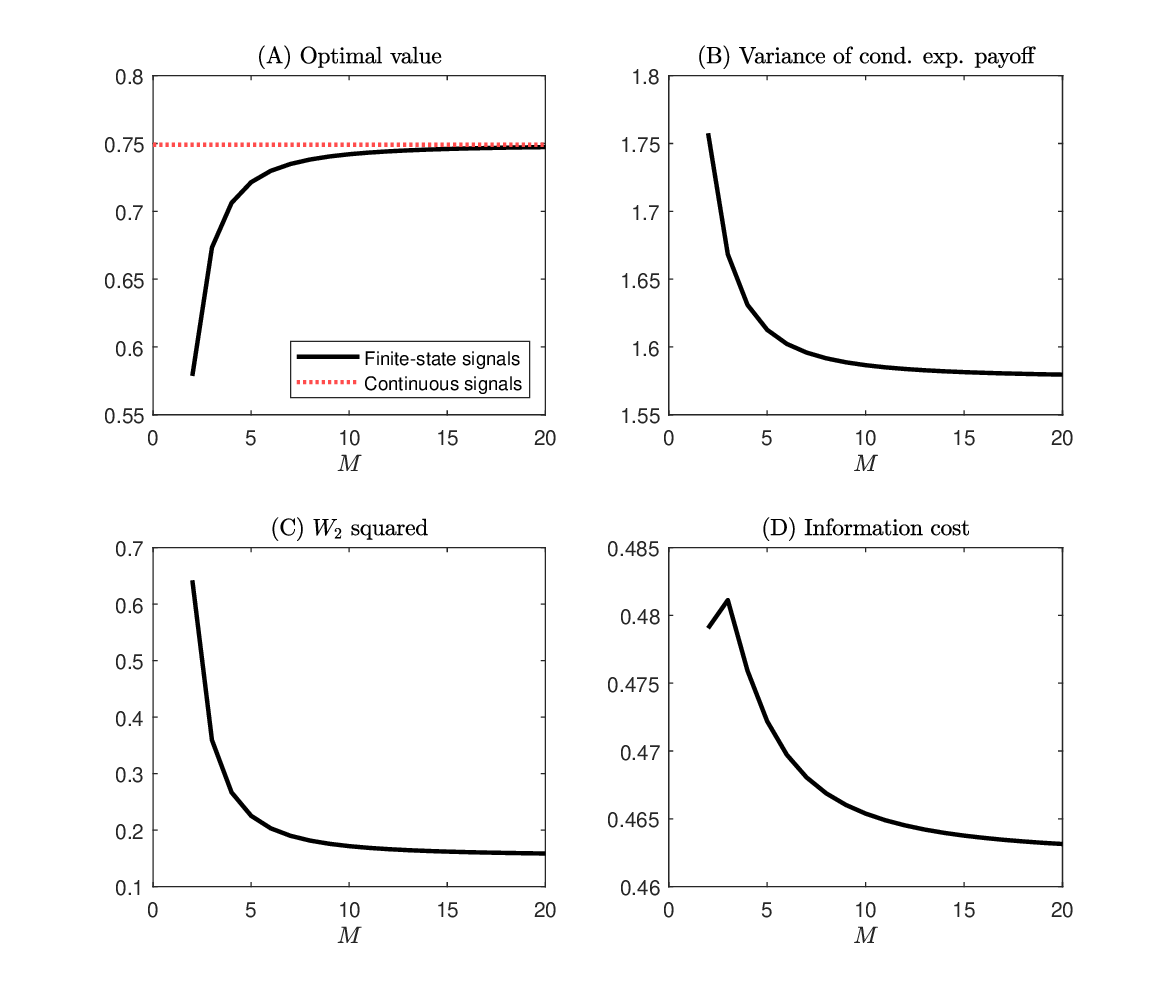}
\caption{Convergence of values (symmetric prior $p(v_1)=0.5$)}

    \begin{minipage}{\textwidth}
	\footnotesize
	\vspace{2mm}
This figure presents the convergence of optimal value and its different components for the informed trader, as the number os signal states $M$ increases.   
Parameters: $N=2$, $v_1 = -2, v_2 = 2$, $p(v_1) = p(v_2) = 0.5$, $\lambda = 2$, $\sigma_Z=1$, and $T=1$.
    \end{minipage}
    \label{fig:conv_sym}
\end{figure}

\begin{figure}[ht]
\centering
\includegraphics[scale=0.6]{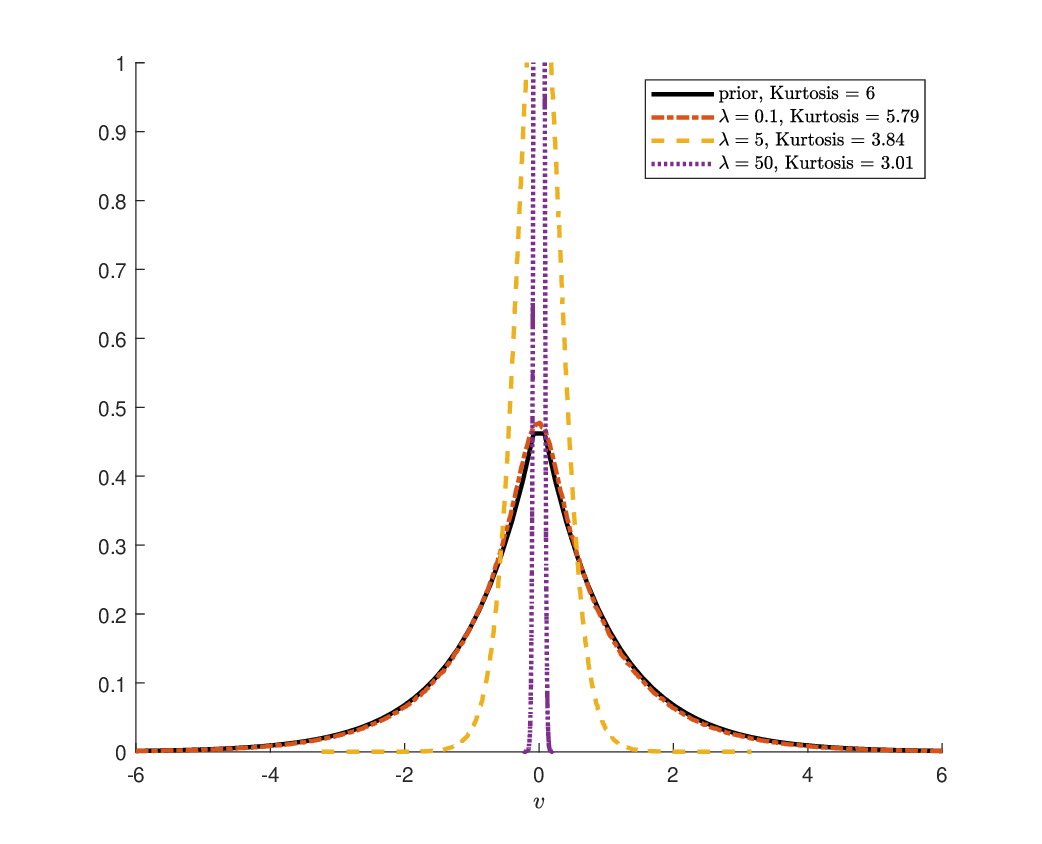}
\caption{Double exponential distributed payoff (different $\lambda$)}
    \begin{minipage}{\textwidth}
	\footnotesize
	\vspace{2mm}
This figure presents the density of the conditional expected payoff and its Kurtosis, when the payoff $\tilde{v}$ has a double exponential distribution.  The signal $\tilde{z}$ has a normal distribution $N(0,\sigma_Z^2 T)$. The density of the prior distribution of $\tilde{v}$ is $\rho(v) = 0.5 e^{-v}$ for $v>0$ and $\rho(v) = 0.5 e^{v}$ for $v<0$. Other parameters are  $\sigma_Z=1$ and $T=1$.
    \end{minipage}
    \label{fig:cont_lam}
\end{figure}

\begin{figure}[ht]
\centering
\includegraphics[scale=0.6]{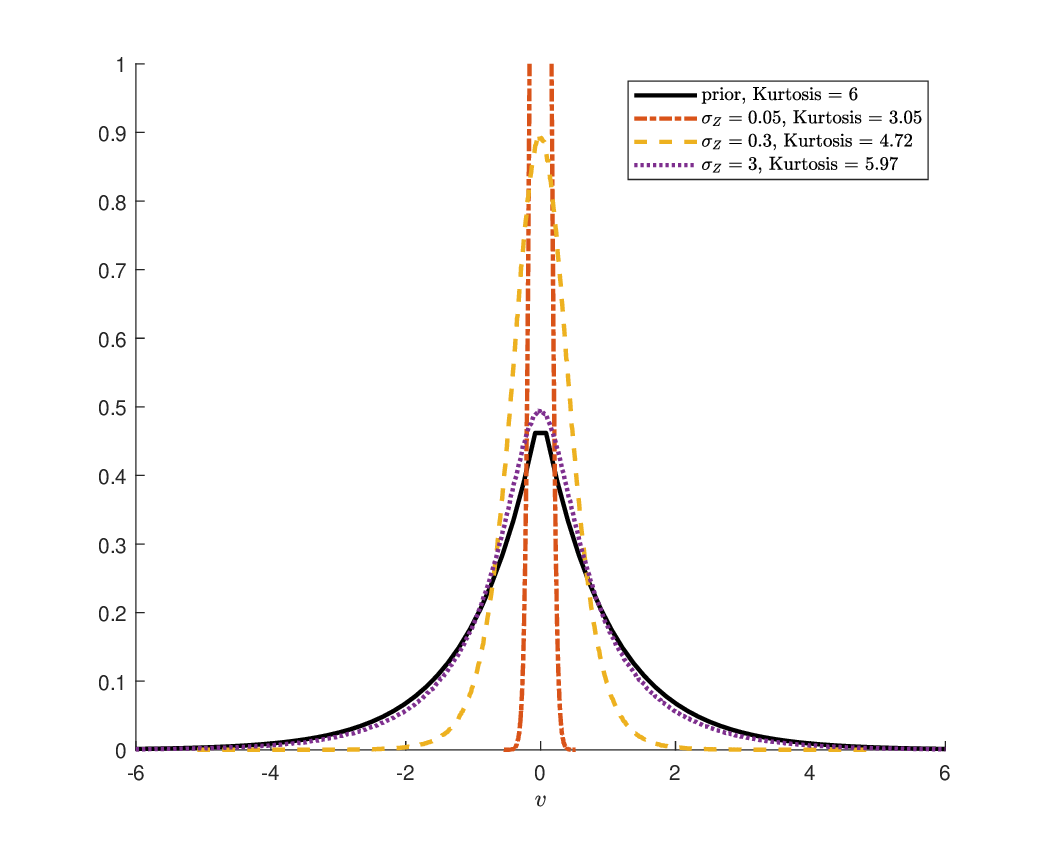}
\caption{Double exponential distributed payoff (different $\sigma$)}
    \begin{minipage}{\textwidth}
	\footnotesize
	\vspace{2mm}
This figure presents the density of the conditional expected payoff and its Kurtosis, when the payoff $\tilde{v}$ has a double exponential distribution.  The signal $\tilde{z}$ has a normal distribution $N(0,\sigma_Z^2 T)$. The density of the prior distribution of $\tilde{v}$ is $\rho(v) = 0.5 e^{-v}$ for $v>0$ and $\rho(v) = 0.5 e^{v}$ for $v<0$. Other parameters are  $\lambda=1$ and $T=1$.
    \end{minipage}
    \label{fig:cont_sig}
\end{figure}


\bibliographystyle{aea}
\bibliography{bibliography}

@article{Sims1998,
  author = {Sims, Christopher A.},
  title = {Stickiness},
  journal = {Carnegie-Rochester Conference Series on Public Policy},
  year = {1998},
  volume = {49},
  number = {1},
  pages = {317--356},
  month = {December},
  publisher = {Elsevier}
}

@article{Sims2003,
  title={Implications of rational inattention},
  author={Sims, Christopher A. },
  journal={Journal of Monetary Economics},
  year={2003},
  volume={50},
  pages={665-690},
}

@article{Kyle1985,
   title = {Continuous auctions and insider trading},
    author = {Kyle, ALbert S.},
    journal = {Econometrica},
    year = {1985},
    volume = {53},
    number = {6},
    pages = {1315--1335},
}

@article{Back1992,
    title = {Insider trading incontinuous time},
    author = {Back, Kerry},
    journal = {Review of Financial Studies},
    year = {1992},
    volume = {5},
    number = {3},
    pages = {387--409},
}

@unpublished{Back_et_al_2021,
    title = {Optimal transport and risk aversion in {K}yle's model of infomred trading},
    author = {Back, Kerry and Cocquemas, Francois and Ekren, Ibrahim and Lioui, Abraham},
    note = {Working paper, Rice University, Florida State University, EDHEC Business School},
    year = {2021},
}

@unpublished{Han2021,
    author = {Han, Brandon},
    title = {Dynamic information acquisition and asset prices},
    note = {Working paper, Smith School, University of Maryland},
    year = {2021},
}

@article{Gelbrich1990,
  title={{On a Formula for the L2 Wasserstein Metric between Measures on Euclidean and Hilbert Spaces}},
  author={Gelbrich, Matthias},
  journal={Mathematische Nachrichten},
  volume={147},
  pages={185-203},
  year={1990},
}

@article{Sinkhorn1964,
  author  = {Sinkhorn, Richard},
  title   = {A Relationship Between Arbitrary Positive Matrices and Doubly Stochastic Matrices},
  journal = {The Annals of Mathematical Statistics},
  year    = {1964},
  volume  = {35},
  number  = {2},
  pages   = {876--879},
}

@article{SinkhornKnopp1967,
  author  = {Sinkhorn, Richard and Knopp, Paul},
  title   = {Concerning nonnegative matrices and doubly stochastic matrices},
  journal = {Pacific Journal of Mathematics},
  year    = {1967},
  volume  = {21},
  number  = {2},
  pages   = {343--348},
}

@article{RuschendorfThomsen1993,
  author  = {R{\"u}schendorf, Ludger and Thomsen, Wolfgang},
  title   = {Note on the {S}chr{\"o}dinger equation and {$I$}-projections},
  journal = {Statistics \& Probability Letters},
  year    = {1993},
  volume  = {17},
  number  = {5},
  pages   = {369--375},
}

@article{PeyreCuturi2019,
  author  = {Peyr{\'e}, Gabriel and Cuturi, Marco},
  title   = {Computational Optimal Transport},
  journal = {Foundations and Trends in Machine Learning},
  year    = {2019},
  volume  = {11},
  number  = {5-6},
  pages   = {355--607},
  doi     = {10.1561/2200000073},
  publisher = {now Publishers}
}

@article{Ruschendorf1995,
  author  = {R{\"u}schendorf, Ludger},
  title   = {Convergence of the {I}terative {P}roportional {F}itting {P}rocedure},
  journal = {The Annals of Statistics},
  year    = {1995},
  volume  = {23},
  number  = {4},
  pages   = {1160--1174},
}

@article{CaplinDeanLeahy2022,
  title = {Rationally Inattentive Behavior: Characterizing and Generalizing Shannon Entropy},
  author = {Caplin, Andrew and Dean, Mark and Leahy, John},
  journal = {Journal of Political Economy},
  volume = {130},
  number = {6},
  pages = {1676--1715},
  year = {2022},
  publisher = {University of Chicago Press},
  doi = {10.1086/719276}
}

@article{CaplinDeanLeahy2019,
  title = {Rational Inattention, Optimal Consideration Sets, and Stochastic Choice},
  author = {Caplin, Andrew and Dean, Mark and Leahy, John},
  journal = {The Review of Economic Studies},
  volume = {86},
  number = {3},
  pages = {1061--1094},
  year = {2019},
  month = {05},
  doi = {10.1093/restud/rdy037},
  publisher = {Oxford University Press}
}

@techreport{CaplinDean2013,
  title = {Behavioral Implications of Rational Inattention with Shannon Entropy},
  author = {Caplin, Andrew and Dean, Mark},
  institution = {National Bureau of Economic Research},
  type = {Working Paper},
  number = {19318},
  year = {2013},
  month = {August},
  doi = {10.3386/w19318},
  url = {https://www.nber.org}
}

@article{MatejkaMcKay2015,
  title = {Rational Inattention to Discrete Choices: A New Foundation for the Multinomial Logit Model},
  author = {Mat\v{e}jka, Filip and McKay, Alisdair},
  journal = {American Economic Review},
  volume = {105},
  number = {1},
  pages = {272--298},
  year = {2015},
  month = {January},
  doi = {10.1257/aer.20130047},
  url = {https://www.aeaweb.org}
}

@article{MackowiakMatejkaWiederholt2023,
  title = {Rational Inattention: A Review},
  author = {Ma\'{c}kowiak, Bartosz and Mat\v{e}jka, Filip and Wiederholt, Mirko},
  journal = {Journal of Economic Literature},
  volume = {61},
  number = {1},
  pages = {226--273},
  year = {2023},
  month = {March},
  doi = {10.1257/jel.20211524},
  url = {https://www.aeaweb.org/articles?id=10.1257/jel.20211524}
}

@book{galichon2016optimal, title={Optimal Transport Methods in Economics}, author={Galichon, Alfred}, year={2016}, publisher={Princeton University Press}, address={Princeton, NJ}, isbn={9780691172767}, url={https://press.princeton.edu/books/hardcover/9780691172767/optimal-transport-methods-in-economics}
}

@incollection{kyle1984market,
  author    = {Kyle, Albert S.},
  title     = {Market Structure, Information, Futures Markets, and Price Formation},
  booktitle = {International Agricultural Trade: Advanced Readings in Price Formation, Market Structure, and Price Instability},
  editor    = {Storey, Gary G. and Schmitz, Andrew and Sarris, Alexander H.},
  pages     = {45--64},
  year      = {1984},
  publisher = {Westview Press},
  address   = {Boulder, CO}
}

@article{admati1988theory,
  title={A Theory of Intraday Patterns: Volume and Price Variability},
  author={Admati, Anat R. and Pfleiderer, Paul},
  journal={The Review of Financial Studies},
  volume={1},
  number={1},
  pages={3--40},
  year={1988},
  publisher={Oxford University Press}
}

@article{vives2014possibility,
  title={On the Possibility of Informationally Efficient Markets},
  author={Vives, Xavier},
  journal={Journal of the European Economic Association},
  volume={12},
  number={5},
  pages={1200--1239},
  year={2014},
  publisher={Oxford University Press}
}

@techreport{khorramithereze2025magic,
  title={The Magic of Markets: Information Acquisition and Aggregation with Strategic Traders},
  author={Khorrami, Paymon and Thereze, Jo\~{a}o},
  institution={Duke University},
  type={Working Paper},
  year={2025},
  url={http://papers.paymonkhorrami.com/Information_Aggregation_in_Financial_Markets.pdf}
}

@article{lambert2018strategic,
  title={Strategic Trading in Informationally Complex Environments},
  author={Lambert, Nicolas S. and Ostrovsky, Michael and Panov, Mikhail},
  journal={Econometrica},
  volume={86},
  number={4},
  pages={1119--1157},
  year={2018},
}

@article{bagnoli2001existence,
  title={On the Existence of Linear Equilibria in Models of Market Making},
  author={Bagnoli, Mark and Viswanathan, S. and Holden, Craig},
  journal={Mathematical Finance},
  volume={11},
  number={1},
  pages={1--31},
  year={2001},
}

@article{BackPedersen1998,
  title = {Long-lived information and intraday patterns},
  author = {Back, Kerry and Pedersen, Hal},
  journal = {Journal of Financial Markets},
  volume = {1},
  number = {3-4},
  pages = {385--402},
  year = {1998},
}

@article{Banerjee2020,
  author = {Banerjee, Snehal and Breon-Drish, Bradyn},
  title = {Strategic trading and unobservable information acquisition},
  journal = {Journal of Financial Economics},
  year = {2020},
  volume = {138},
  number = {2},
  pages = {458--482},
}

@article{Banerjee2022,
  author = {Banerjee, Snehal and Breon-Drish, Bradyn},
  title = {Dynamics of Research and Strategic Trading},
  journal = {Review of Financial Studies},
  year = {2022},
  volume = {35},
  number = {2},
  pages = {908--961},
}

@article{BackBaruch2004,
  author = {Back, Kerry and Baruch, Shmuel},
  title = {Information in Securities Markets: Kyle Meets Glosten and Milgrom},
  journal = {Econometrica},
  year = {2004},
  volume = {72},
  number = {2},
  pages = {433--465},
}

@article{CollinDufresneFos2016,
  author = {Collin-Dufresne, Pierre and Fos, Vyacheslav},
  title = {Insider Trading, Stochastic Liquidity, and Equilibrium Prices},
  journal = {Econometrica},
  year = {2016},
  volume = {84},
  number = {4},
  pages = {1441--1475},
}

@article{CaldenteyStacchetti2010,
  author = {Caldentey, René and Stacchetti, Ennio},
  title = {Insider Trading With a Random Deadline},
  journal = {Econometrica},
  year = {2010},
  volume = {72},
  number = {1},
  pages = {245--283},
  doi = {10.3982/ECTA7884},
  url = {https://doi.org/10.3982/ECTA7884}
}

@article{GrossmanStiglitz1980,
  author = {Grossman, Sanford J. and Stiglitz, Joseph E.},
  title = {On the Impossibility of Informationally Efficient Markets},
  journal = {The American Economic Review},
  year = {1980},
  volume = {70},
  number = {3},
  pages = {393--408},
  publisher = {American Economic Association},
  url = {https://www.aeaweb.org}
}

@article{Hellwig1980,
  author = {Hellwig, Martin F.},
  title = {On the Aggregation of Information in Competitive Markets},
  journal = {Journal of Economic Theory},
  year = {1980},
  volume = {22},
  number = {3},
  pages = {477--498},
  doi = {10.1016/0022-0531(80)90056-3},
  url = {https://doi.org/10.1016/0022-0531(80)90056-3}
}

@article{BreonDrish2015,
  author = {Breon-Drish, Bradyn},
  title = {On Existence and Uniqueness of Equilibrium in a Class of Noisy Rational Expectations Models},
  journal = {The Review of Economic Studies},
  year = {2015},
  volume = {82},
  number = {3},
  pages = {868--921},
  doi = {10.1093/restud/rdv001},
  url = {https://doi.org}
}

@article{Chabakauri2022,
  author = {Chabakauri, Georgy and Yuan, Kathy and Zachariadis, Konstantinos E.},
  title = {Multi-asset Noisy Rational Expectations Equilibrium with Contingent Claims},
  journal = {The Review of Economic Studies},
  year = {2022},
  volume = {89},
  number = {5},
  pages = {2445--2490},
  doi = {10.1093/restud/rdab081},
  url = {https://doi.org}
}

@article{Detemple2020,
  author = {Detemple, Jerome and Rindisbacher, Marcel and Robertson, Scott},
  title = {Dynamic Noisy Rational Expectations Equilibrium With Insider Information},
  journal = {Econometrica},
  year = {2020},
  volume = {88},
  number = {6},
  pages = {2697--2737},
  doi = {10.3982/ECTA17038},
  url = {https://doi.org}
}

@Article{Miao_Xing:2024,
    author = {Jianjun Miao and Hao Xing},
    title = {Dynamic discrete choice under rational inattention},
    journal = {Economic Theory},
    volume = {77},
    pages = {597–652},
    year = {2024},
}

@unpublished{Woodford:2012,
    author = {Michael Woodford},
    title = {Inattentive valuation and reference-dependent choice},
    Institution = {Columbia University},
    Note = {Working Paper},
    year = {2012},
}

@article{Dean_Neligh:2023,
    author = {Mark Dean and Nathaniel Neligh},
    title = {Experimental tests of rational inattention},
    journal = {Journal of Political Economy},
    Volume = {131},
    Number = {12},
    Year = {2023},
    pages = {3415-3461},
}

@article{Dean_Neligh:2020,
    author = {Mark Dean and Nathaniel Neligh},
    title = {Estimating information cost functions in models of rational inattention},
    Journal = {Journal of Economic Theory},
    Volume = {187},
    pages = {1-32},
    year = {2020},
}

@article{Hebert_Woodford:2021,
    author = {Benjamin H\'{e}bert and Michael Woodford},
    title = {Neighborhood-based information costs},
    journal = {American Economic Review},
    volume = {111},
    number = {10},
    pages = {3225-3255},
    year = {2021},
}

@article{Pomatto:2023,
    author = {Luciano Pomatto and Philipp Strack and Omer Tamuz},
    title = {The cost of information: the case of constant marginal costs},
    journal = {American Economic Review},
    volume = {113},
    number = {5},
    pages = {1360-1393},
    year = {2023},
}

@inproceedings{justiniano2025entropy,
  author    = {Justiniano, Jorge and Kleiner, Andreas and Moldovanu, Benny and Rumpf, Martin and Strack, Philipp},
  title     = {Entropy-Regularized Optimal Transport in Information Design},
  booktitle = {Proceedings of the 26th ACM Conference on Economics and Computation (EC '25)},
  year      = {2025},
  month     = {July},
  doi       = {10.1145/3736252.3742617},
  publisher = {ACM}
}

\newpage

\setcounter{section}{0}
\renewcommand{\thesection}{APPENDIX \Alph{section}}
\renewcommand{\theequation}{\Alph{section}.\arabic{equation}}
\renewcommand{\thefigure}{\Alph{section}.\arabic{figure}}
\renewcommand{\thetable}{\Alph{section}.\arabic{table}}

\setcounter{equation}{0} \setcounter{figure}{0}
\setcounter{table}{0}
\appendix
\numberwithin{equation}{section}

\section{Proofs}

\subsection{Proof of Lemma \ref{lem:cond_exp}}

Given that $\mathbb{E}[\tilde{v}|\tilde{s}]$ is measurable with respect to $\sigma(Y_T)$, the definition of conditional expectation, that it minimizes the mean squared error, implies 
\begin{equation}
    \label{proof:cond_exp_ineq1}
    \mathbb{E}\big[\big(\tilde{v} - \mathbb{E}[\tilde{v}|\tilde{s}]\big)^2\big] \geq \mathbb{E}\big[\big(\tilde{v} - \mathbb{E}[\tilde{v}|\mathcal{F}^Y_T] \big)^2\big].
\end{equation}
Meanwhile, because $\tilde{v}$ and $Z$ are independent, $\mathbb{E}\big[\tilde{v}| \mathcal{F}^Z_T \vee \sigma (\tilde{s})\big] = \mathbb{E}[\tilde{v}|\tilde{s}]$. Then $\mathcal{F}^Y_T \subset \mathcal{F}^Z_T \vee \sigma (\tilde{s})$ yields 
\begin{equation}
    \label{proof:cond_exp_ineq2}
    \mathbb{E}\big[\big(\tilde{v} - \mathbb{E}[\tilde{v}| \mathcal{F}^Y_T]\big)^2\big] \geq \mathbb{E}\big[\big(\tilde{v} - \mathbb{E}[\tilde{v}| \mathcal{F}^Z_T\vee \sigma(\tilde{s})]\big)^2\big] = \mathbb{E}\big[\big(\tilde{v} - \mathbb{E}[\tilde{v}| \tilde{s}]\big)^2\big].
\end{equation}
Combining \eqref{proof:cond_exp_ineq1} and \eqref{proof:cond_exp_ineq2} yields
\[
 \mathbb{E}\big[\big(\tilde{v} - \mathbb{E}[\tilde{v}|\tilde{s}]\big)^2\big]  =  \mathbb{E}\big[\big(\tilde{v} - \mathbb{E}[\tilde{v}| \mathcal{F}^Y_T]\big)^2\big].
\]
Due to the uniqueness of conditional expectation, we conclude $\mathbb{E}[\tilde{v}|\tilde{s}] = \mathbb{E}\big[\tilde{v}|\mathcal{F}^Y_T\big]$.

\subsection{Proof of Proposition \ref{prop:corres}}

Consider an equilibrium in Definition \ref{def:equilibrium}. Because the informed agent's trading strategy is measurable with respect to $\sigma(\tilde{s})$, we have seen from the discussion leading to \eqref{def:insider_obj_trans} that 
\begin{equation}\label{proof:equal_exp_pro}
  \mathbb{E} \Big[\int_0^T \big(\mathbb{E}[\tilde{v}] - P_t\big) \theta_t dt \Big] =  \mathbb{E} \Big[\int_0^T \big(\mathbb{E}[\tilde{v}|\tilde{s}] - P_t\big) \theta_t dt \Big].
\end{equation}
Therefore, the informed trader's strategy $(p(\cdot|\cdot), q(\cdot))$ and $\theta_t$ also maximizes the objective in \eqref{def:insider_obj_trans}. Meanwhile, for a Markovian pricing rule, \cite[Lemma 2]{Back1992} shows that the optimal trading strategy satisfies 
\begin{equation}\label{proof:P_T_id}
P_T = H(T,Y_T) = \mathbb{E}[\tilde{v}|\tilde{s}].
\end{equation}
Therefore,
\[
P_t = \mathbb{E}\big[\tilde{v}\,|\, \mathcal{F}^Y_t\big] = \mathbb{E} \Big[\mathbb{E}\big[\tilde{v}| \mathcal{F}^Y_T\big]\, \Big| \,\mathcal{F}^Y_t\Big] = \mathbb{E} \Big[\mathbb{E}\big[\tilde{v}| \tilde{s}\big]\, \Big| \,\mathcal{F}^Y_t\Big], 
\]
where the second equality follows from the law of iterated conditional expectation, the third equality holds due to Lemma \ref{lem:cond_exp}. Therefore, the pricing rule satisfies \eqref{def:price_2}, and the equilibrium satisfies the requirement in Definition \ref{def:equilibrium_2}.

Conversely, consider an equilibrium in Definition \ref{def:equilibrium_2}. Because the informed agent's trading strategy is measurable with respect to $\sigma(\tilde{s})$, Equation \eqref{proof:equal_exp_pro} still holds. Hence the informed trader's strategy $(p(\cdot|\cdot), q(\cdot))$ and $\theta_t$ also maximizes the objective in \eqref{def:insider_obj}. Meanwhile, \cite[Lemma 2]{Back1992} implies \eqref{proof:P_T_id} for a Markovian pricing rule. Therefore, Lemma \ref{lem:cond_exp} implies that
\[
P_t = \mathbb{E} \Big[\mathbb{E}\big[\tilde{v}| \tilde{s}\big]\, \Big| \,\mathcal{F}^Y_t\Big] = \mathbb{E} \Big[\mathbb{E}\big[\tilde{v}| \mathcal{F}^Y_T\big]\, \Big| \,\mathcal{F}^Y_t\Big] = \mathbb{E}\big[\tilde{v}\,|\, \mathcal{F}^Y_t\big].
\]
Therefore, this equilibrium also satisfies all requirements in Definition \ref{def:equilibrium}.

\subsection{Proof of Proposition \ref{prop:normal}}
The first order condition of \eqref{signal_str_pro_normal} in $\xi$ is 
\begin{equation}\label{proof:FOC_xi}
 \xi^{-\frac12} = \frac{\lambda}{\sigma_v \sigma_Z \sqrt{T}} \frac{1}{1-\xi}.
\end{equation}
The left-hand side of the equation above decreases from infinity to zero and the right-hand side increases from $\lambda$ to infinity, as $\xi$ increases from 0 to 1. Therefore, there exists a unique solution $\xi^*\in(0,1)$ to this first order condition. The objective function in \eqref{signal_str_pro_normal} is concave in $\xi$, because the second order derivative of the objective function is negative. Therefore, $\xi^*$ is the unique maximizer. 

When $\frac{\lambda}{\sigma_v \sigma_Z \sqrt{T}}$ increases, the right-hand side of \eqref{proof:FOC_xi} increases for a fixed $\xi$. Therefore, $\xi^*$ is smaller in order to satisfy the first-order condition \eqref{proof:FOC_xi}. 

Given $\xi^*$, $\sigma^2_{\epsilon}$ satisfies $\xi^* = \frac{\sigma_v^2}{\sigma^2_{v} + \sigma^2_{\epsilon}}$. Hence, the optimal conditional precision of the signal satisfies \eqref{normal_opt_precision}. Because the right-hand side of \eqref{normal_opt_precision} increases with $\xi^*$, the observation from the previous paragraph implies that the optimal signal precision $\frac{1}{\sigma^2_{\epsilon}}$ decreases with $\frac{\lambda}{\sigma_v \sigma_Z \sqrt{T}}$.

\subsection{Proof of Proposition \ref{prop:normal_opt}}

We divide the proof into several steps. First, we add another constraint to the problem \eqref{signal_str_pro} that the variance of $\mathbb{E}[\tilde{v}|\tilde{s}]$ is a given constant $K$. We show that the optimal value of the problem \eqref{signal_str_pro} with the additional constraint is bounded from above by the optimal value of an alternative problem. Second, we show that a normal distribution $(p(\cdot|\cdot), q(\cdot))$ achieves the optimal value for this alternative problem, and it also achieves the optimal value for the original problem. Third, we allow the informed trader to vary $K$ and finally confirm the statement in Proposition \ref{prop:normal_opt}. 

\medskip

\noindent\underline{Step 1}:  Consider the problem \eqref{signal_str_pro} with the constraints \eqref{con:q_int}, \eqref{con:p_int}, and \eqref{con:Bayes}. We consider an additional constraint that the variance of $\mathbb{E}[\tilde{v}|\tilde{s}]$ is a given constant $K$. Equivalently, the second moment of $\mathbb{E}[\tilde{v}|\tilde{s}]$ is given by
\begin{equation}\label{con:second_mom}
 \mathbb{E}\big[\mathbb{E}[\tilde{v}\, |\, \tilde{s}]^2\big] = \int \Big(\int v \,p(v|s) dv \Big)^2 q(s) ds = K + \mathbb{E}[\tilde{v}]^2.
\end{equation}

We introduce an alternative problem whose optimal value provides an upper bound for the problem \eqref{signal_str_pro} with the constraints \eqref{con:q_int}, \eqref{con:p_int}, \eqref{con:Bayes}, and \eqref{con:second_mom}. To this end, recall that $\mathbb{E}[\tilde{v}|\tilde{s}]$ and $Z_T$ are independent. For an arbitrary CDF $F_s$ for $\mathbb{E}[\tilde{v}|\tilde{s}]$, its Wasserstein-2 distance admits the following \cite{Gelbrich1990} lower bound : 
\begin{equation}\label{Gelbrich_bd}
 W_2^2(F_s, G) \geq (\mu_G - \mathbb{E}[\tilde{v}])^2 + (\sigma_G - \sqrt{K})^2,
\end{equation}
where $\mu_G = 0$ is the mean of $Z_T$ and $\sigma_G = \sigma_Z \sqrt{T}$ is the standard deviation. Moreover, the lower bound on the right-hand side of \eqref{Gelbrich_bd} is attained when $F_s$ is a normal CDF with mean $\mathbb{E}[\tilde{v}]$ and variance $K$. Let us now replace $W_2^2(F_s, G)$ in \eqref{signal_str_pro} by its lower bound and consider the following problem:
\begin{align}
&\max_{p(\cdot|\cdot), q(\cdot)} \frac12  \mathbb{E}\big[\mathbb{E}[\tilde{v}| \tilde{s}]^2\big]   -\lambda I(\tilde{s}; \tilde{v}) + \frac12 \sigma_Z^2 T - \frac12 \mathbb{E}[\tilde{v}]^2 - \frac12 \big(\sigma_G - \sqrt{K} \big)^2,\label{pro_2}\\
\text{subject to } & \eqref{con:q_int}, \eqref{con:p_int}, \eqref{con:Bayes}, \text{ and } \eqref{con:second_mom}. \nonumber
\end{align}
The value of problem \eqref{pro_2} provides an upper bound to the value of problem \eqref{signal_str_pro}. Moreover, if we show the value of \eqref{pro_2} is attained by a normal distribution $(p^*(\cdot|\cdot), q^*(\cdot))$, then $F_s$ is a normal CDF and the Gelbrich lower bound in \eqref{Gelbrich_bd} is attained. Hence $(p^*(\cdot|\cdot), q^*(\cdot))$ also attains the optimal value of problem \eqref{signal_str_pro}.

\medskip

\noindent \underline{Step 2}: We consider the problem \eqref{pro_2} with the constraints \eqref{con:q_int}, \eqref{con:p_int}, \eqref{con:Bayes}, and \eqref{con:second_mom}. To this end, recall from \eqref{def:info_cost} that 
\begin{align*}
    \lambda I(\tilde{s}; \tilde{v}) = &\lambda \big[\mathbb{H}(p) - \int \mathbb{H}(p(\cdot |s)) q(s) ds\big]\\
    = & \lambda\iint \log\big(p(v|s)\big) p(v|s) q(s) \, dv ds - \lambda\int \log(p(v)) p(v) \, dv,
\end{align*}
Plugging the previous expression of $\lambda I(\tilde{s}; \tilde{v})$ into \eqref{pro_2} and omitting quantities independent of $p(\cdot|\cdot)$ and $q(\cdot)$, we obtain the following problem:
\begin{align}
&\max_{p(\cdot|\cdot), q(\cdot)}  \frac12 \int \Big(\int v p(v|s) dv \Big)^2 q(s) \, ds - \lambda \iint \log\big(p(v|s)\big) p(v|s) q(s) \, dv ds\label{pro_3}\\
\text{subject to } & \eqref{con:q_int}, \eqref{con:p_int}, \eqref{con:Bayes}, \text{ and } \eqref{con:second_mom}. \nonumber
\end{align}
Introducing Lagrangian multiplier for the four constraints above, we obtain the following unconstrained problem
\begin{equation}\label{pro:uncon}
\begin{split}
&\max_{p(\cdot|\cdot), q(\cdot)} \frac12 \int \Big(\int v p(v|s) dv \Big)^2 q(s) \, ds - \lambda \iint \log\big(p(v|s)\big) p(v|s) q(s) \, dv ds \\
&+ \int \mu(v)\Big(\int p(v|s) q(s) ds - p(v)\Big) dv + \tau \Big(\int \Big(\int v p(v|s) dv \Big)^2 q(s) ds - K - \mathbb{E}[\tilde{v}]^2\Big) \\
& + \int \nu(s) \Big(\int p(v|s) dv -1 \Big) ds + \eta \Big(\int q(s) ds -1 \Big).
\end{split}
\end{equation}
The objective function above is linear in $q(s)$. Hence, for $s$ with $q^*(s)\neq 0$, the following first-order condition must be satisfied at optimal
\begin{equation}
    \label{FOC:q}
    \begin{split}
0 = \frac12 \Big(\int v p^*(v|s) dv\Big)^2 - \lambda \int \log (p^*(v|s)) p^*(v|s) dv + \int \mu(v) p^*(v|s) dv + \tau \Big(\int v p^*(v|s) dv \Big)^2 + \eta.
    \end{split}
\end{equation}
For $s$ with $q^*(s)\neq 0$,
the first-order condition in $p(v|s)$ yields 
\begin{equation}
    \label{FOC:p}
    0 =  (1+ 2\tau) v \, \int v' p^*(v'|s) dv'  - \lambda - \lambda \log p^*(v|s) + \mu(v) + \frac{\nu(s)}{q(s)}.
\end{equation}
Because the marginal value of $-\log\big(p(v|s)\big) p(v|s)$ is infinite at $p(v|s) =0$, the optimizer $p^*(v|s)$ must be nonzero except for some $v$'s with measure zero. Therefore \eqref{FOC:p} must be satisfied for almost all $v$.

The equation \eqref{FOC:p} implies 
\begin{equation}
    \label{p:exp}
    p^*(v|s) = M(s) e^{(1+2\tau) v \,\mathbb{E}[\tilde{v}|\tilde{s} = s]/\lambda + \mu(v)/\lambda}, 
\end{equation}
for some function $M(s)$ which ensures the constraint $\int p^*(v|s) dv =1$ to be satisfied for any $s$.
Taking log on both sides of \eqref{p:exp}, multiplying both sides by $p^*(v|s)$, and integrating with respect to $v$, we obtain 
\begin{equation*}
    \int \log(p^*(v|s)) p^*(v|s) dv = \log M(s) \int p^*(v|s) dv + \frac{1+2\tau}{\lambda} \big(\mathbb{E}[\tilde{v} |\tilde{s}= s] \big)^2 + \frac{1}{\lambda} \int \mu(v) p^*(v|s) dv. 
\end{equation*}
For the function $M(s)$ which ensures the constraint $\int p^*(v|s) dv =1$ for any $s$, we obtain from the previous equation that 
\begin{equation*}
    \lambda \int \log(p^*(v|s)) p^*(v|s) dv  - \int \mu(v) p^*(v|s) dv = \lambda \log M(s) + (1+2\tau) \mathbb{E}[\tilde{v}|\tilde{s}= s]^2.
\end{equation*}
Plugging the previous equation back into \eqref{FOC:q}, we obtain 
\[
M(s) = e^{-(1/2 + \tau) \mathbb{E}[\tilde{v} | \tilde{s}=s]^2/\lambda + \eta /\lambda}. 
\]
Combining the previous equation with \eqref{p:exp}, we obtain 
\begin{equation}
    \label{p:exp_2}
    p^*(v|s) = e^{\frac{1/2+\tau}{\lambda} 2 v \mathbb{E}[\tilde{v} |\tilde{s}= s] - \frac{1/2+\tau}{\lambda} \mathbb{E}[\tilde{v}|\tilde{s}= s]^2 + \frac{\mu(v)+ \eta}{\lambda}} = e^{\frac{\mu(v) +\eta}{\lambda} + \frac{1/2+\tau}{\lambda} v^2} \, e^{-\frac{1/2+\tau}{\lambda} \big(v - \mathbb{E}[\tilde{v} |\tilde{s}= s]\big)^2}.
\end{equation}
Denote $\sigma_{v|s} := \sqrt{\frac{\lambda}{1+ 2\tau}}$. Lemma \ref{lem:normal_posterior} below implies that $p^*(v|s)$ must be a normal density with mean $\mathbb{E}[\tilde{v}|\tilde{s}= s]$ and standard deviation $\sigma_{v|s}$. Hence $\mu(v)$ must satisfy
\begin{equation}
    \label{eq:mu}
    e^{\frac{\mu(v) -\eta}{\lambda} + \frac{1/2+\tau}{\lambda} v^2} = \frac{1}{\sqrt{2\pi} \sqrt{\frac{\lambda}{1+2\tau}}}, \quad \text{ for any } v.
\end{equation}
Given that both the posterior density $p^*(v|s)$ and the prior density $p(v)$ are normal, the marginal density $q^*(s)$ must be chosen to satisfy the Bayes constraint  $\int p^*(v|s) q^*(s) ds = p(v)$ and optimize the objective function in \eqref{pro:uncon}. Plugging the densities for $p^*(v|s)$ and $p(v)$, the Bayes constraint is rewritten as 
\begin{equation}\label{normal_bayesian}
    \int \frac{1}{\sqrt{2\pi} \sigma_{v|s}} e^{- \frac{(v- \mathbb{E}[\tilde{v} |\tilde{s}= s])^2}{2\sigma^2_{v|s}}} q^*(s) ds = \frac{1}{\sqrt{2\pi} \sigma_v} e^{-\frac{(v- \mathbb{E}[\tilde{v}])^2}{2\sigma_v^2}}, \quad \text{ for any } v.
\end{equation}
This constraint can be satisfied by choosing a normal density $\widetilde{q}^*$ with mean $\mathbb{E}[\tilde{v}]$ and variance $\sigma^2_v/\xi$, with $\xi = 1- \frac{\sigma_{v|s}^2}{\sigma_v^2}$. Meanwhile, note that the objective function in \eqref{pro:uncon} is linear in $q(s)$, hence \eqref{FOC:q} implies that $\widetilde{q}^*$ does not change the value of the objective function. This choice of $\tilde{q}^*$ also ensures that $\tilde{q}^*(s) \neq 0$ for any $s$, verifying a condition used in \eqref{FOC:p}. Therefore, $(p^*(\cdot|\cdot), \widetilde{q}^*(\cdot))$ is an optimizer for \eqref{pro:uncon}, and both $p^*(\cdot|\cdot)$ and $\widetilde{q}^*(\cdot)$ are normal.

To complete this step, we note that both $\sigma_{v|s}$ and $\xi$ are determined by $\tau$, which is the Langrangian multiplier associated to the constraint \eqref{con:second_mom}. For a given $\tau$, the distribution $(p^*(\cdot|\cdot), \tilde{q}^*(\cdot))$ can be implemented by introducing a single
\begin{equation}\label{proof:signal_impl}
\tilde{s} = \tilde{v} + \epsilon,
\end{equation}
where $\epsilon\sim N(0, \sigma^2_{\epsilon})$ with $\sigma^2_{\epsilon} = \sigma_v^2 (\frac{1}{\xi}-1)$. It follows from \eqref{proof:signal_impl} that $\mathbb{E}[\tilde{v} | \tilde{s}= s] = \mathbb{E}[\tilde{v}] + \xi (s-\mathbb{E}[\tilde{v}])$.
Therefore, 
\begin{align*}
  K=\text{Var}\big(\mathbb{E}[\tilde{v}|\tilde{s}]\big)  = \xi^2 \text{Var}(\tilde{s}) = \xi \sigma_v^2. 
\end{align*}
When $\tau$ increases, $\sigma_{v|s}$ decreases and $\xi$ increases. Hence, there is an 1-to-1 monotone increasing relation between $\tau$ and $K$. 

\medskip
\noindent\underline{Step 3}: As the informed trader varies $K$ to maximize the objective in \eqref{pro_2}, there is a unique associated $\tau$ and the optimal $(p^*(\cdot|\cdot), \tilde{q}^*(\cdot))$ are both normal and their variances are determined by $\tau$, or equivalently, by $\xi$. As we have seen in the proof of Proposition \ref{prop:normal}, there exists a unique $\xi^*$ which maximizes the objective in \eqref{signal_str_pro_normal}. Meanwhile, the associated optimal $(p^*(\cdot|\cdot), \tilde{q}^*(\cdot))$ is normal, hence the Gelbrich lower bound is attained, hence 
$(p^*(\cdot|\cdot), \tilde{q}^*(\cdot))$ is also an optimizer for the problem \eqref{signal_str_pro} with the constraints \eqref{con:q_int}, \eqref{con:p_int}, and \eqref{con:Bayes}.

\subsection{An auxiliary lemma for the proof of Proposition \ref{prop:normal_opt}}

\begin{lemma}\label{lem:normal_posterior}
Suppose that $\frac{d}{ds} \mathbb{E}[\tilde{v}|\tilde{s}=s] \neq 0$ for any $s$ with $q(s) \neq 0$. Then the density $p(v|s)$, which satisfies 
\begin{equation}
    \label{p_normal}
    p(v|s) = e^{\frac{\mu(v) +\eta}{\lambda} + \frac{1/2+\tau}{\lambda} v^2} \, e^{-\frac{1/2+\tau}{\lambda} \big(v - \mathbb{E}[\tilde{v} |\tilde{s} = s]\big)^2}, \quad \text{for any } v \text{ and } s,
\end{equation}
must be a normal density with mean $\mathbb{E}[\tilde{v}|\tilde{s} = s]$ and standard deviation $\sigma_{v|s}=\sqrt{\frac{\lambda}{1+ 2\tau}}$.
\end{lemma}
\begin{proof}
 It follows from the definition of conditional expectation that $\int v p(v|s) dv = \mathbb{E}[\tilde{v}|\tilde{s} = s]$. Plugging the expression of $p(v|s)$ from \eqref{p_normal} into the previous equation and taking derivative with respect to $s$ on both sides, we obtain
 \[
 \int v p(v|s) \frac{v-\mathbb{E}[\tilde{v}|s]}{\sigma^2_{v|s}} \frac{d}{ds} \mathbb{E}[\tilde{v}|\tilde{s} = s] \, dv = \frac{d}{ds}\mathbb{E}[\tilde{v}|\tilde{s} = s].
 \]
 Because $\frac{d}{ds} \mathbb{E}[\tilde{v}|\tilde{s} = s]\neq 0$, the previous equation and $\int v p(v|s) dv = \mathbb{E}[\tilde{v}|\tilde{s} = s]$ imply that 
 \begin{equation}
     \label{p_second}
     \int \big(v - \mathbb{E}[\tilde{v}|\tilde{s} = s]\big)^2 p(v|s) dv = \sigma^2_{v|s}.
 \end{equation}
 Taking derivatives with respect to $s$ on both sides of \eqref{p_second}, because $\sigma^2_{v|s}= \frac{\lambda}{2+2\tau}$ is independent of $s$, we obtain
 \[
 -2 \frac{d}{ds} \mathbb{E}[\tilde{v}|\tilde{s} = s] \int \big(v - \mathbb{E}[\tilde{v}\, |\, s]\big) p(v|s) dv + \frac{d}{ds} \mathbb{E}[\tilde{v}|\tilde{s} = s] \int \big(v- \mathbb{E}[\tilde{v}|\tilde{s} = s] \big)^3 \frac{1}{\sigma^2_{v|s}} p(v|s) dv =0.
 \]
 This previous equation, together with $\int v p(v|s) dv = \mathbb{E}[\tilde{v}|\tilde{s} = s]$ and $\frac{d}{ds} \mathbb{E}[\tilde{v}|\tilde{s} = s] \neq 0$, implies that 
 \begin{equation}
     \label{p_third}
     \int \big(v- \mathbb{E}[\tilde{v}|\tilde{s} = s] \big)^3 p(v|s) dv =0.
 \end{equation}
 Taking the derivative with respect to $s$ to both sides of the equation above and using the similar argument as before, we obtain 
 \begin{equation}
     \label{p_fourth}
     \int \big(v- \mathbb{E}[\tilde{v}|\tilde{s} = s] \big)^4 p(v|s) dv = 3 \sigma^2_{v|s}.
 \end{equation}

 Using induction and similar argument as before, we prove that 
 \[
 \int \big(v- \mathbb{E}[\tilde{v}|\tilde{s} = s] \big)^{2n-1} p(v|s) dv = 0 \quad \text{and} \quad \int \big(v- \mathbb{E}[\tilde{v}|\tilde{s} = s] \big)^{2n} p(v|s) dv = \sigma_{v|s}^{2n} \prod_{k=1}^n(2k-1),
 \]
 for any positive integer $n$. These are exactly the central moments of a normal distribution with mean $\mathbb{E}[\tilde{v}|\tilde{s} = s]$ and variance $\sigma_{v|s}^2$. The statement follows, because moments uniquely identify a distribution.
\end{proof}

\subsection{Proof of Lemma \ref{lem:pro_cont}}
We claim that 
\begin{equation}
    \label{Fs_quant_cont}
    F^{-1}_s\big(Q(s)\big) = \mathbb{E}[\tilde{v}|s], \quad \text{for any } s.
\end{equation}
To see this, suppose that $Q(s) = F_s(v)$ for some $v$. Then 
\begin{equation}\label{proof_lem_cont_1}
Q(s) = F_s(v) = \mathbb{P}\big(\mathbb{E}[\tilde{v}|\tilde{s}] \leq v \big).
\end{equation}
Meanwhile,
\begin{equation}\label{proof_lem_cont_2}
Q(s) = \mathbb{P}(\tilde{s} \leq s) = \mathbb{P}\big( \mathbb{E}[\tilde{v}|\tilde{s}] \leq \mathbb{E}[\tilde{v}|s]\big),
\end{equation}
where the second equality holds because $s\mapsto \mathbb{E}[\tilde{v}|s]$ is strictly increasing. Comparing \eqref{proof_lem_cont_1} and \eqref{proof_lem_cont_2}, we obtain $v = \mathbb{E}[\tilde{v}|s]$, and hence \eqref{Fs_quant_cont} is confirmed.

Plug \eqref{Fs_quant_cont} into \eqref{expected_profit_cont} and use change-of-variable $u=Q(s)$, 
\[
\int_0^1 F_s^{-1}(u) G^{-1}(u) \, du = \int F_s^{-1} \big(Q(s) \big) G^{-1}\big( Q(s)\big) q(s) \, ds = \int \mathbb{E}[\tilde{v}|s] G^{-1}\big( Q(s)\big) q(s) \, ds.
\]
The information acquisition cost is 
\[
\lambda I (\tilde{s}; \tilde{v}) = \lambda \Big\{\sum_{n=1}^N \int \log \big(p(v_n|s) \big) p(v_n|s) q(s) \, ds - \sum_{n=1}^N \log \big( p(v_n)\big) p(v_n) \Big\}.
\]
Combining the previous two equations, we obtain the objective function in \eqref{obj_cont}. The constraint \eqref{con:p_cont} ensures that the posterior $p(\cdot|s)$ is a probability for any $s$ and the constraint \eqref{con:Baye_cont} is the Bayes plausibility.

\subsection{Proof of Proposition \ref{prop:opt_cont}}
\noindent\underline{Statements in (i)}:
Introduce the Lagrangian multiplier $\mu_n$ and $\nu(s)$, $1\leq n\leq N, s\in \text{support}(q)$, for the constraints \eqref{con:p_cont} and \eqref{con:Baye_cont}. We obtain the following unconstrained problem for \eqref{obj_cont}:
\begin{align}
\min_{\mu, \nu} \max_{p(\cdot|\cdot), q(\cdot)}  & \int \mathbb{E}[\tilde{v}| s] G^{-1}\big(Q(s)\big) q(s) \nonumber \\
&- \lambda \Big\{\int \sum_{n=1}^N \log \big(p(v_n|s) \big) p(v_n| s) q(s) \, ds -\sum_{n=1}^N \log \big(p(v_n)\big) p(v_n)\Big\} \nonumber\\
 &+ \int \nu(s) \big[\sum_{n=1}^N p(v_n | s) -1\big] \,ds +  \sum_{n=1}^N \mu_n \big[\int p(v_n | s) q(s)ds - p(v_n)\big]. \label{obj_uncon_cont}
\end{align}
Recall that $\mathbb{E}[\tilde{v}| s] = \sum_{n=1}^N v_n p(v_n | s)$. When $q(s) \neq 0$, the first-order condition of $p(v_n|s)$ in \eqref{obj_uncon_cont} yields
\[
v_n G^{-1}\big(Q(s)\big) q(s) - \lambda q(s) - \lambda \log \big(p(v_n|s)\big) q(s) + \nu(s) + \mu_n q(s) =0.
\]
This implies that 
\[
p(v_n | s) = C(s) \, e^{\frac{1}{\lambda}(v_n G^{-1}(Q(s)) +\mu_n)},
\]
for some factor $C(s)$. The constraint $\sum_{n=1}^N p(v_n | s) =1$ pins down the factor $C(s)$ and leading to the expression 
\begin{equation}\label{opt_p_cont_proof}
p(v_n|s) = \frac{e^{\frac{1}{\lambda}(v_n G^{-1}(Q(s)) +\mu_n)}}{\sum_{n'=1}^N e^{\frac{1}{\lambda}(v_{n'} G^{-1}(Q(s)) +\mu_{n'})}}.
\end{equation}
Given that the objective function in \eqref{obj_uncon_cont} is concave in $p(v_n|s)$ due to the entropy cost, $p(v_n| s)$ is also the maximizer for given $q(s)$ and $(\mu_1, \dots, \mu_N)$.

Given $p(v_n |s)$ and $q(s)$ for $1\leq n\leq N$ and $s\in \text{support}(q)$, the existence and uniqueness (up to an additive constant) of $(\mu_1, \dots, \mu_N)$ has been discussed in the text. Such choice of $\mu$'s ensures the Bayes plausibility $\int p(v_n| s) q(s)\, ds = p(v_n)$ for any $n$.

Plugging \eqref{opt_p_cont_proof} into \eqref{obj_uncon_cont}, we obtain the following objective function for $q$:
    \begin{equation}\label{obj:q_cont_proof}
    \begin{split}
        \lambda \int \log \Big(\sum_{n=1}^N e^{\frac{1}{\lambda}\big(v_n G^{-1}(Q(s)) + \mu_n\big)} \Big) q(s) \, ds + \lambda \sum_{n=1}^N \log\big(p(v_n)\big) p(v_n) - \sum_{n=1}^N \mu_n p(v_n) .
    \end{split}
    \end{equation}

\smallskip

\noindent\underline{Statements in (ii)}: For any continuously distributed $\tilde{s}$ with the cumulative distribution function $Q$, $Q(\tilde{s})$ is uniformly distributed on $[0,1]$. Hence $G^{-1}(Q(\tilde{s}))$ has a normal distribution $N(0, \sigma_Z \sqrt{T})$. Therefore, 
\begin{align}
    \int \log \Big(\sum_{n=1}^N e^{\frac{1}{\lambda}\big(v_n G^{-1}(Q(s)) + \mu_n\big)} \Big) q(s) \, ds = & \mathbb{E} \Big[\log \Big(\sum_{n=1}^N e^{\frac{1}{\lambda}\big(v_n G^{-1}(Q(\tilde{s})) + \mu_n\big)} \Big)\Big] \nonumber\\
    = & \mathbb{E} \Big[\log \Big(\sum_{n=1}^N e^{\frac{1}{\lambda}\big(v_n \tilde{z} + \mu_n\big)} \Big)\Big], \label{exp_profit_cont_proof}
\end{align}
where $\tilde{z}$ is a normal random variable with distribution $N(0,\sigma_Z^2 T)$. Plugging \eqref{exp_profit_cont_proof} into \eqref{obj:q_cont_proof}, we obtain \eqref{opt_value:cont}. 
For the constraint \eqref{con:Baye_cont}, it can be stated equivalently as 
\begin{equation}\label{proof:pvn}
\begin{split}
p(v_n) = & \int p(v_n|s) q(s) ds = \int \frac{e^{\frac{1}{\lambda}(v_n G^{-1}(Q(s)) +\mu_n)}}{\sum_{n'=1}^N e^{\frac{1}{\lambda}(v_{n'} G^{-1}(Q(s)) +\mu_{n'})}} q(s) \, ds \\
=&  \mathbb{E} \Big[\frac{e^{\frac{1}{\lambda}(v_n G^{-1}(Q(\tilde{s})) +\mu_n)}}{\sum_{n'=1}^N e^{\frac{1}{\lambda}(v_{n'} G^{-1}(Q(\tilde{s})) +\mu_{n'})}}\Big]
=  \mathbb{E} \Big[\frac{e^{\frac{1}{\lambda}(v_n \tilde{z} +\mu_n)}}{\sum_{n'=1}^N e^{\frac{1}{\lambda}(v_{n'} \tilde{z} +\mu_{n'})}}\Big].
\end{split}
\end{equation}
The previous identities show that the constraint \eqref{con:Baye_cont} is the same for any continuously distributed $\tilde{s}$. Therefore, the constraint \eqref{con:Baye_cont} is satisfied by choosing $(\mu_1, \dots, \mu_N)$ as the solution of the Sinkhorn problem associated with \eqref{proof:pvn}, which exists and is unique (up to additive constants); see \cite{RuschendorfThomsen1993}. Finally, note that the value in \eqref{proof:pvn} does not change if $(\mu_1, \dots, \mu_N)$ is replaced by $(\mu_1 + c, \dots, \mu_N+c)$ for any $c$. The statement in (ii) is confirmed.

\smallskip

\noindent{\underline{Statements in (iii)}}: It follows from \eqref{opt_p_cont_proof} that 
\[
\mathbb{E}[\tilde{v}|\tilde{s}] = \sum_{n=1}^N v_n p(v_n|\tilde{s}) = \frac{\sum_{n=1}^N v_n e^{\frac{1}{\lambda}(v_n G^{-1}(Q(\tilde{s})) +\mu_n)}}{\sum_{n=1}^N e^{\frac{1}{\lambda}(v_{n} G^{-1}(Q(\tilde{s})) +\mu_{n})}} \sim \frac{\sum_{n=1}^N v_n e^{\frac{1}{\lambda}(v_n \tilde{z} +\mu_n)}}{\sum_{n=1}^N e^{\frac{1}{\lambda}(v_{n} \tilde{z} +\mu_{n})}}.
\]
In the equation above $\sim$ means that two random variables share the same distribution. The equal distribution above holds due to the fact that $G^{-1}(Q(\tilde{s}))$ and $\tilde{z}$ share the same distribution.  Hence, any continuous signal leads to the same expected payoff distribution. In particular, the informed trader can choose $\tilde{z}$ as the signal. 

Finally, we verify that $z\mapsto \mathbb{E}[\tilde{v}|z]$ is strictly increasing. To this end, denote $w_n := \frac{e^{(v_n z + \mu_n)/\lambda}}{\sum_{n=1}^N e^{(v_n z+\mu_n)/\lambda}}$. Note that $\sum_{n=1}^N w_n =1$ and $\mathbb{E}[\tilde{v}| z] = \sum_{n=1}^N v_n w_n$. Then 
\[
\frac{d}{dz} \mathbb{E}[\tilde{v}|z] = \sum_{n=1}^n v_n^2 w_n - \Big(\sum_{n=1}^N v_n w_n \Big)^2 >0,
\]
where the inequality follows from Cauchy-Schwarz inequality.

\subsection{Proof of Corollary \ref{cor:two_state_density}}
In this case, $N=2$. For any $v\in (v_1, v_2)$, we compute the CDF for $\mathbb{E}[\tilde{v}|\tilde{z}]$:
\begin{align*}
 F_s(v) =& \mathbb{P}\big(\mathbb{E}[\tilde{v}|\tilde{z}] \leq v \big) = \mathbb{P} \Big(\sum_{n=1}^2 v_n e^{(v_n \tilde{z} +\mu_n)/\lambda} \leq v \sum_{n=1}^2 e^{(v_n \tilde{z} + \mu_n)/\lambda}\Big) \\
 =& \mathbb{P}\Big((v_1 - v) e^{(v_1 \tilde{z} + \mu_1)/\lambda} \leq (v-v_2) e^{(v_2 \tilde{z} +\mu_2)/\lambda}\Big) = \mathbb{P}\left(\tilde{z} \leq \frac{\lambda}{v_1 - v_2}\log \Big(\frac{v_2-v}{v-v_1} \Big) - \frac{\mu_1 - \mu_2}{v_1 - v_2} \right) \\
 = & \mathbb{P}\left(\frac{\tilde{z}}{\sigma_Z \sqrt{T}} \leq \frac{\lambda}{\sigma_Z \sqrt{T}(v_1 - v_2)}\log \Big(\frac{v_2-v}{v-v_1} \Big) - \frac{\mu_1 - \mu_2}{\sigma_Z \sqrt{T}(v_1 - v_2)} \right)
\end{align*}
where the second equality follows from the form of $\mathbb{E}[\tilde{v}|\tilde{z}]$ in \eqref{exp_payoff:cont} and $\frac{\tilde{z}}{\sigma_Z \sqrt{T}}$ has a standard normal distribution. Taking the derivative with respect to $v$ in $F_s(v)$, we obtain the density in \eqref{two_state_density}.

\subsection{Proof of Proposition \ref{prop:disc_gen}}

Introduce the Lagrangian multiplier $\nu_m$ and $\mu_n$, $1\leq n\leq N, 1\leq m\leq M$, for the constraints \eqref{con:p_prob} and \eqref{con:bayesian}. We obtain the following unconstrained problem for \eqref{insider_prob_disc}:
\begin{align}
\min_{\mu, \nu} \max_{p(\cdot|\cdot), q(\cdot)}  & \sum_{m=1}^M \mathbb{E}[\tilde{v}\,|\, s_m] I_m q(s_m) \nonumber \\
&- \lambda \Big\{\sum_{m=1}^M \sum_{n=1}^N \log \big(p(v_n\,|\,s_m) \big) p(v_n\,|\, s_m) q(s_m) -\sum_{m=1}^M \log \big(p(v_n)\big) p(v_n)\Big\} \nonumber\\
 &+ \sum_{m=1}^M \nu_m \Big[\sum_{n=1}^N p(v_n\,|\, s_m) -1\Big]+  \sum_{n=1}^N \mu_n \Big[\sum_{m=1}^M p(v_n\,|\, s_m) q(s_m) - p(v_n)\Big]. \label{obj_uncon_disc}
\end{align}
Recall that $\mathbb{E}[\tilde{v}\,|\, s_m] = \sum_{n=1}^N v_n p(v_n\,|\, s_m)$. When $q(s_m) \neq 0$, the first-order condition of $p(v_n\,|\, s_m)$ in \eqref{obj_uncon_disc} yields
\[
v_n I_m q(s_m) - \lambda q(s_m) - \lambda \log \big(p(v_n\,|\, s_m)\big) q(s_m) + \nu_m + \mu_n q(s_m) =0.
\]
This implies that 
\[
p(v_n \,|\, s_m) = C_m e^{(v_n I_m +\mu_n)/\lambda},
\]
for some factor $C_m$. The constraint $\sum_{n=1}^N p(v_n\,|\, s_m) =1$ pins down the factor $C_m$ and leading to the expression 
\begin{equation}\label{opt_p_disc_proof}
p(v_n\,|\, s_m) = \frac{e^{(v_n I_m +\mu_n)/\lambda}}{\sum_{n'=1}^N e^{(v_{n'} I_m +\mu_{n'})/\lambda}}.
\end{equation}
Given that the objective function in \eqref{obj_uncon_disc} is concave in $p(v_n\,|\, s_m)$ due to the entropy cost, $p(v_n\,|\, s_m)$ is also the maximizer for given $q$'s and $\mu$'s.

Given $p(v_n \,|\, s_m)$ and $q(s_m)$ for $1\leq n\leq N, 1\leq m\leq M$, the existence and uniqueness (up to an additive constant) of $(\mu_1, \dots, \mu_N)$ is the consequence of the (rectangular) Sinkhorn-Knopp problem. Here rectangular means that $M\neq N$. This is the discrete analogue of the (continuous) Sinkhorn-Knopp problem in Proposition \ref{prop:opt_cont}.   Such choice of $\mu$'s ensures the Bayes plausibility $\sum_{m=1}^M p(v_n\,|\, s_m) q(s_m) = p(v_n)$ for any $n$.

Plugging \eqref{opt_p_disc_proof} into \eqref{obj_uncon_disc}, we obtain the objective function for $q$'s in \eqref{obj:q_disc}. Note that for signal state $s_m$ with $q(s_m)=0$, it does not contribute to the objective function in \eqref{obj:q_disc}. Finally, we show that there exists a maximizer of the objective function \eqref{obj_uncon_disc} on the simplex of $\mathbb{R}^M$. To this end, note that $\mathbb{E}[\tilde{v}|s_m]$ is bounded from above by the maximum value of $\{v_1, \dots, v_N\}$. Meanwhile, 
\[
\big|I_m q_m\big| = \Big|\int_{Q_{m-1}}^{Q_m} G^{-1}(d) du\Big| \leq \int_0^\infty G^{-1}(u) du = \frac{\sigma_Z\sqrt{T}}{\sqrt{2\pi}},
\]
where the last equality holds because $G$ is the CDF of $N(0,\sigma_Z^2 T)$. Meanwhile, the information acquisition cost in \eqref{insider_prob_disc} is nonnegative. Therefore, the objective function in \eqref{insider_prob_disc} is uniformly bounded from above. Consequently, the objective function \eqref{obj:q_disc} is bounded from above on the simplex of $\mathbb{R}^{M}$. The objective function is also continuous. Hence there exists an optimizer.

\subsection{Proof of Proposition \ref{prop:cont_optimal}}

If the signal has $M$ different signal states, we denote the value in \eqref{obj_uncon_disc} as $V^M(p)$, where $p = (p(v_1), \dots, p(v_N))$ is the prior payoff distribution. Equation \eqref{obj_uncon_disc} implies that the value function $V^M(p)$ can be rewritten as
\begin{align}
    V^M(p) = &\min_{\mu} \Big\{\widehat{V}^M(\mu) - \sum_{n=1}^N \mu_n p(v_n) \Big\}, \quad \text{where } \mu = (\mu_1, \dots, \mu_N) \text{ and }\label{duality} \\
    \widehat{V}^M(\mu) := &\min_{\nu} \max_{p(\cdot|\cdot), q(\cdot)} \sum_{m=1}^M \mathbb{E}[\tilde{v}\,|\, s_m] I_m q(s_m) \nonumber \\
&- \lambda \Big\{\sum_{m=1}^M \sum_{n=1}^N \log \big(p(v_n\,|\,s_m) \big) p(v_n\,|\, s_m) q(s_m) -\sum_{n=1}^N \log \big(p(v_n)\big) p(v_n)\Big\} \nonumber\\
 &+ \sum_{m=1}^M \nu_m \Big[\sum_{n=1}^N p(v_n\,|\, s_m) -1\Big]+  \sum_{n=1}^N \mu_n \sum_{m=1}^M p(v_n\,|\, s_m) q(s_m).
\end{align}

For each fixed $\mu$, we claim that 
\begin{equation}\label{hatV_ineq}
 \widehat{V}^M(\mu) < \widehat{V}^{M+1}(\mu).
\end{equation}
Then, subtract $\sum_{n=1}^N \mu_n p(v_n)$ on both sides,
\[
\widehat{V}^M(\mu) - \sum_{n=1}^N \mu_n p(v_n) < \widehat{V}^{M+1}(\mu) - \sum_{n=1}^N \mu_n p(v_n).
\]
Taking infimum over $\mu$ first on the left-hand side, then on the right-hand side, we obtain that 
\begin{equation}
 V^M(p) \leq V^{M+1}(p).
\end{equation}

It then remains to prove the claim \eqref{hatV_ineq}. To this end, it follows from Proposition \ref{prop:disc_gen} that 
\[
\widehat{V}^M(\mu) = \max_{q(\cdot)} \Big\{  \lambda \sum_{m=1}^M \log \Big(\sum_{n=1}^N e^{(v_n I_m + \mu_n)/\lambda} \Big) q(s_m)  \Big\} + \lambda \sum_{n=1}^N \log\big(p(v_n)\big) p(v_n). 
\]
Similar argument as in the proof of Proposition \ref{prop:disc_gen} shows that the objective function in $q$ is uniformly bounded from above on the simplex of $\mathbb{R}^M$. Meanwhile, this objective function is continuous. Therefore, an optimizer exists and we call it $q^M$. 

Next, we are going to construct a new signal with $M+1$ signal states which achieves strictly higher value than $q^M$. To achieve this goal, let us consider a signal state $s_m$ with the associated probability $q^M(s_m)$. Split the probability $q^M(s_m)$ into $\overline{q_m}$ and $\underline{q_m}$ such that $\overline{q_m} + \underline{q_m} = q^M(s_m)$. Replace the signal state $s_m$ by two new signal states $\overline{s_m}$ and $\underline{s_m}$, meanwhile, keep all other signal states and associated probabilities the same. Given that all $q^M(s_{m'})$ remain the same for $m'\neq m$, let us compare the following terms affected by the splitted signal:
\[
\log \Big(\sum_{n=1}^N e^{(v_n I_m +\mu_n )/\lambda} \Big) q(s_m) \quad \text{and } \quad \log \Big(\sum_{n=1}^N e^{(v_n \underline{I_m} +\mu_n)/\lambda} \Big) \underline{q_m} + \log \Big(\sum_{n=1}^N e^{(v_n \overline{I_m}+\mu_n)/\lambda} \Big) \overline{q_m},
\]
where $\underline{I_m} := \frac{1}{\underline{q_m}} \int_{Q_{m-1}}^{Q_{m-1} + \underline{q_m}} G^{-1}(u) du$ and $\overline{I_m}:= \frac{1}{\overline{q_m}} \int_{Q_m - \overline{q_m}}^{Q_m} G^{-1}(u) du$. Denote $\alpha := \frac{\underline{q_m}}{q(s_m)}$, then $1-\alpha = \frac{\overline{q_m}}{q(s_m)}$. Observe that 
\[
\alpha \underline{I_m} + (1-\alpha) \overline{I_m} = I_m.
\]
Consider the function 
\[
f(x) : = \log \Big(\sum_{n=1}^N e^{(v_n x+\mu_n)/\lambda}\Big).
\]
We have that 
\[
f^{''}(x) = \sum_{n=1}^N (v_n)^2 w_n - \Big(\sum_{n=1}^N v_n w_n\Big)^2, \quad w_n = \frac{e^{(v_n x+\mu_n)/\lambda}}{\sum_{n'=1}^N e^{(v_{n'}x + \mu_{n'})/\lambda}}. 
\]
The Cauchy-Schwarz inequality implies that $f^{''}(x)>0$, hence $f$ is strictly convex. It then follows from the Jensen's inequality that 
\[
\begin{split}
&\log \Big(\sum_{n=1}^N e^{(v_n I_m+\mu_n)/\lambda} \Big) q(s_m) =  q(s_m) f(I_m) = q(s_m) f\big( \alpha \underline{I_m} + (1-\alpha) \overline{I_m} \big) \\
&< q(s_m) \alpha f(\underline{I_m}) + q(s_m) (1-\alpha) f(\overline{I_m}) =\log \Big(\sum_{n=1}^N e^{(v_n \underline{I_m}+\mu_n)/\lambda} \Big) \underline{q_m} + \log \Big(\sum_{n=1}^N e^{(v_n \overline{I_m}+\mu_n)/\lambda} \Big) \overline{q_m}. 
\end{split}
\]
Therefore, splitting the signal strictly improves the value. Given that $q^M$ is the optimal marginal when $M$ signal states are allowed, splitting $q(s_m)$  strictly improves the value. Optimizing over signals with $M+1$ states weakly improve the value further, hence \eqref{hatV_ineq} is confirmed.

\subsection{Two-state symmetric case}

\begin{proposition}
    \label{prop:two_state_sym}
    Suppose $N=2$ and $p(v_1)=p(v_2) =0.5$. The Lagrangian multiplier $\mu$ are zero for symmetric signals.
\end{proposition}

\begin{proof}
We consider discrete signal with the number of signal states $M$. The proof for continuous signal is similar. 

Given that $q(s_m) = q(s_{M+1-m})$ for any $m$ and the distribution of $N(0,\sigma_Z^2 T)$ is symmetric around zero, we obtain from the definition of $I_m$ that $I_m = -I_{M+1 -m}$ for any $1\leq m\leq M$. Consider the posterior 
\[
p(v_n | s_m) = \frac{e^{v_n I_m/\lambda}}{e^{v_1 I_m/\lambda} + e^{v_2 I_m/\lambda}}, \quad n = 1, 2.
\]
We claim that this posterior satisfies the Bayesian identity
\begin{equation}\label{Bayes_N_2}
\sum_{m=1}^M p(v_n|s_m) q(s_m) = p(v_n) =0.5, \quad n=1,2. 
\end{equation}
It then follows from the uniqueness of $(\mu_1, \mu_2)$ (up to additive constants) that $\mu_1 = \mu_2 = 0$.

To see why \eqref{Bayes_N_2} holds, we observe that 
\[
\begin{split}
\sum_{m=1}^M p(v_1|s_m) q(s_m) = &\sum_{m=1}^M \frac{q(s_m)}{1 + e^{(v_2 - v_1) I_m/\lambda}} = \sum_{m=1}^M \frac{q(s_{M+1-m})}{1+ e^{(v_1 - v_2) I_{M+1-m} /\lambda}} \\
=& \sum_{m'=1}^M \frac{q(s_{m'})}{1+ e^{(v_1-v_2) I_{m'}/\lambda}} = \sum_{m'=1}^M p(v_2|s_{m'}) q(s_{m'}),
\end{split}
\]
where the second equality follows from $q(s_m) = q(s_{M+1-m})$ and $I_{m} = -I_{M+1-m}$, the third equality holds due to $m'= M+1-m$. Because $p(v_1|s_m) + p(v_2|s_m)=1$ for any $m$. The previous displayed equation implies \eqref{Bayes_N_2}. 
\end{proof}

\subsection{Proof of Proposition \ref{prop:cont_v}}

The proof follows that for Proposition \ref{prop:opt_cont} verbatim, except by replacing all summations over $v_n$ with integrals with respect to $v$. Because $\mathbb{E}[\tilde{v}^2]<\infty$, the objective function in \eqref{signal_str_pro} is bounded from above. Therefore, all objective functions in the proof are bounded from above as well.

\subsection{Sinkhorn algorithm}\label{app:sinkhorn_pos}

The Sinkhorn problem \eqref{Sinkhorn-comp} can be solved by the Sinkhorn algorithm:
\begin{enumerate}
\item[Step 0:] Initalize $\xi^{(0)}_n=1$ for any $n\in\{1,\dots,N\}$;
\item[Step i:] Update $\zeta(z)$ via 
\begin{equation}\label{Sinkhorn:zeta_update}
\zeta^{(i)}(z) = \frac{q(z)}{\sum_{n'=1}^N \xi_{n'}^{(i-1)} k_{n'}(z)},
\end{equation}
Update $\xi_n$ via 
\begin{equation}\label{Sinkhorn:xi_update}
\xi^{(i)}_n = \frac{p_n}{\int k_n(z) \zeta^{(i)}(z) dz}, \quad 1\leq n\leq N.
\end{equation}
Iterate until convergence.
\end{enumerate}
\cite{Ruschendorf1995} prove that this algorithm converges to the unique solution of the Sinkhorn problem \eqref{Sinkhorn-comp}. For the finite state cases, the convergence of the algorithm was proved by \cite{Sinkhorn1964}. See also \cite{PeyreCuturi2019} for a recent survey of the Sinkhorn algorithm. When the payoff is continuous or the signal state is finite, \eqref{Sinkhorn-comp} and \eqref{Sinkhorn:xi_update} by using summations in finite state space cases and integrals in continuous state space cases.

When $p_n$ is larger, $\xi^{(i)}_n$ in \eqref{Sinkhorn:xi_update} increases. Then $\zeta^{(i+1)}(z)$ increases, because the denominator of \eqref{Sinkhorn:zeta_update} decreases. It then follows from \eqref{Sinkhorn:xi_update} for the step $i+1$ that $\xi^{(i+1)}_n$ increases, because the denominator in \eqref{Sinkhorn:xi_update} decreases. Following the previous iteration, the limit $\xi_n$, which is a part of Sinkhorn problem solution, increases. Because $\mu_n = \lambda\log(\xi_n)$, we confirm that $\xi_n$ increases with $p_n$.

\section{Two payoff states with asymmetric prior}\label{app:asy_pri}

Figures \ref{fig:two_state_asy_lam}, \ref{fig:two_state_asy_sig}, and \ref{fig:conv_asy} present examples for two-state asymmetric prior, i.e., $N=2$ and $p(v_1)\neq p(v_2)$. The results are similar to those from the symmetric prior case (see Figures \ref{fig:two_state_sym_lam}, \ref{fig:two_state_sym_sig}, and \ref{fig:conv_sym}). Because the prior distribution is asymmetric, the distributions for $\mathbb{E}[\tilde{v}|\tilde{s}]$ are skewed. However, Panel (D) of Figures \ref{fig:two_state_asy_lam} and \ref{fig:two_state_asy_sig} shows that the skewness of $\mathbb{E}[\tilde{v}|\tilde{s}]$ converges to $0$ (the skewness of a normal distribution) as either the information acquisition cost parameter $\lambda$ approaches infinity or the noise trade volatility $\sigma_Z$ converges to zero.

\begin{figure}[h!]
\centering
\includegraphics[scale=0.8]{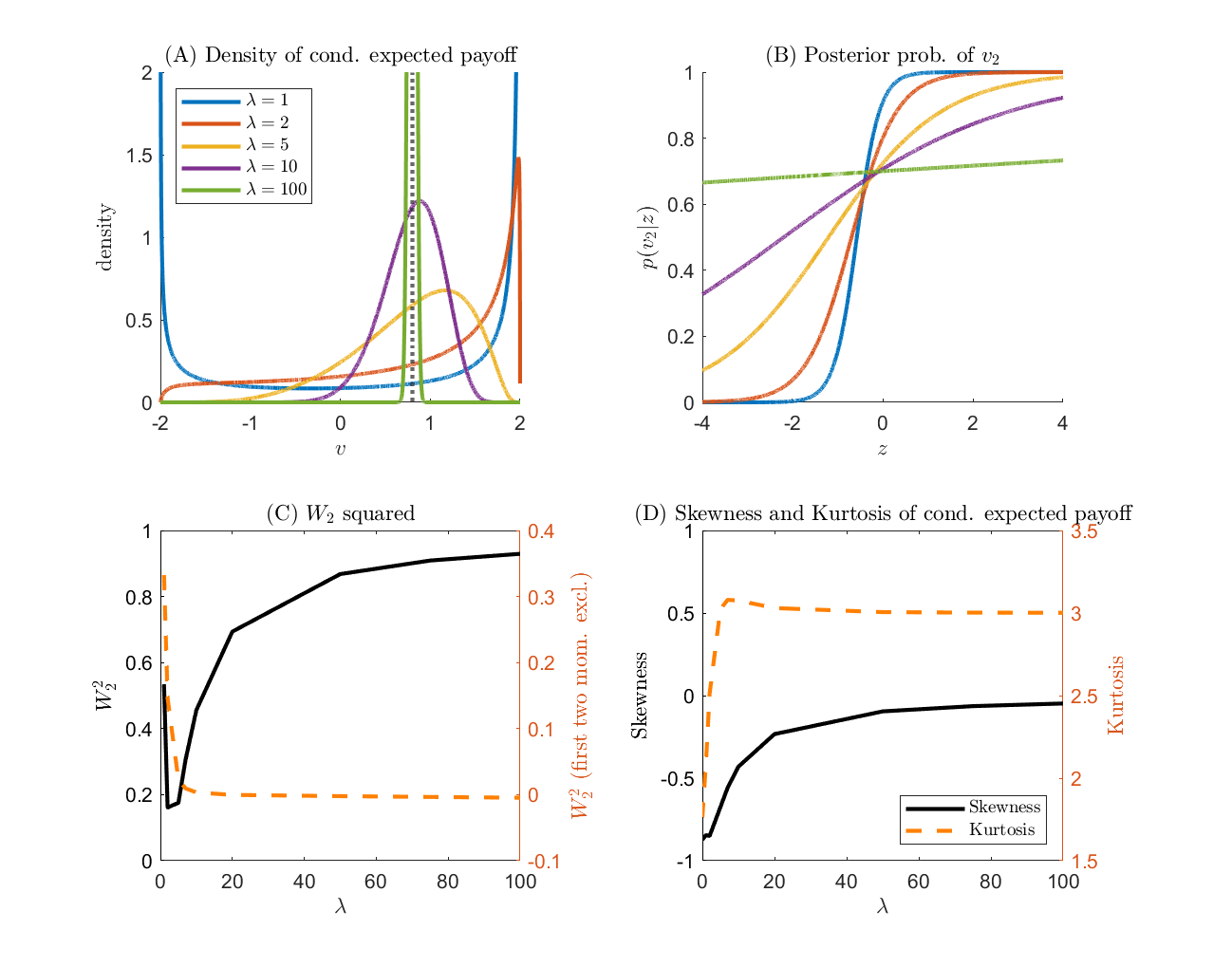}
\caption{Two-state asymmetric prior (different $\lambda$)}
    \begin{minipage}{\textwidth}
	\footnotesize
	\vspace{2mm}
This figure presents the distribution of the conditional expected payoff, its skewness, Kurtosis, its $W_2$-distance with noise trade distribution, and the optimal posterior $p(v_2 | z)$ for the two-state asymmetric prior case. The signal $\tilde{z}$ has a normal distribution $N(0,\sigma_Z^2 T)$.   
Parameters: $N=2$, $v_1 = -2, v_2 = 2$, $p(v_1) = 0.3, p(v_2) = 0.7$, $\sigma_Z=1$, and $T=1$.
    \end{minipage}
    \label{fig:two_state_asy_lam}
\end{figure}

\begin{figure}[h!]
\centering
\includegraphics[scale=0.8]{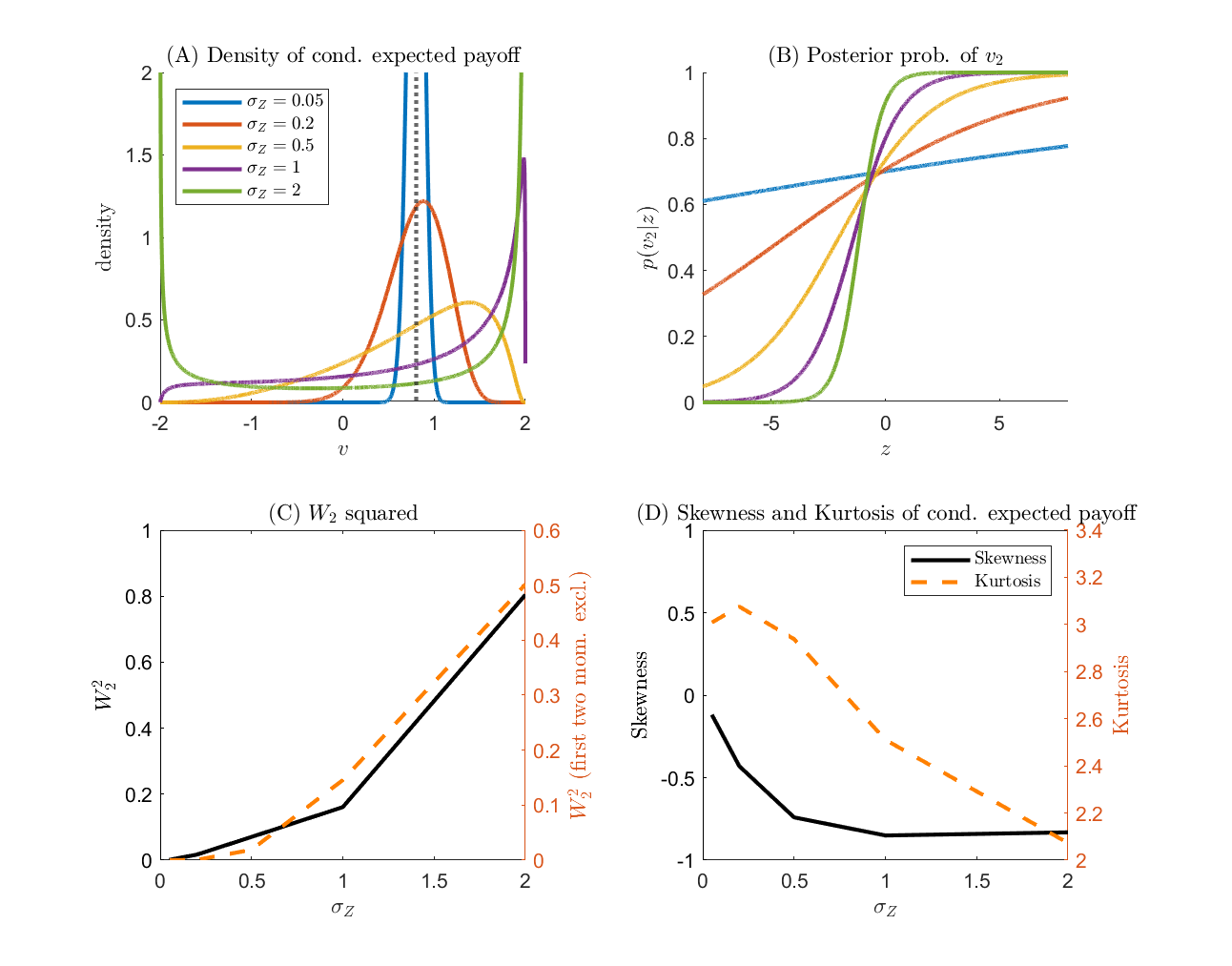}
\caption{Two-state asymmetric prior (different $\sigma_Z$)}
    \begin{minipage}{\textwidth}
	\footnotesize
	\vspace{2mm}
This figure presents the distribution of the conditional expected payoff, its skewness, Kurtosis, its $W_2$-distance with noise trade distribution, and the optimal posterior $p(v_2 | z)$ for the two-state asymmetric prior case. The signal $\tilde{z}$ has a normal distribution $N(0,\sigma_Z^2 T)$.   
Parameters: $N=2$, $v_1 = -2, v_2 = 2$, $p(v_1) = 0.3, p(v_2) = 0.7$, $\lambda=2$, and $T=1$.
    \end{minipage}
    \label{fig:two_state_asy_sig}
\end{figure}

\begin{figure}[h!]
\centering
\includegraphics[scale=0.7]{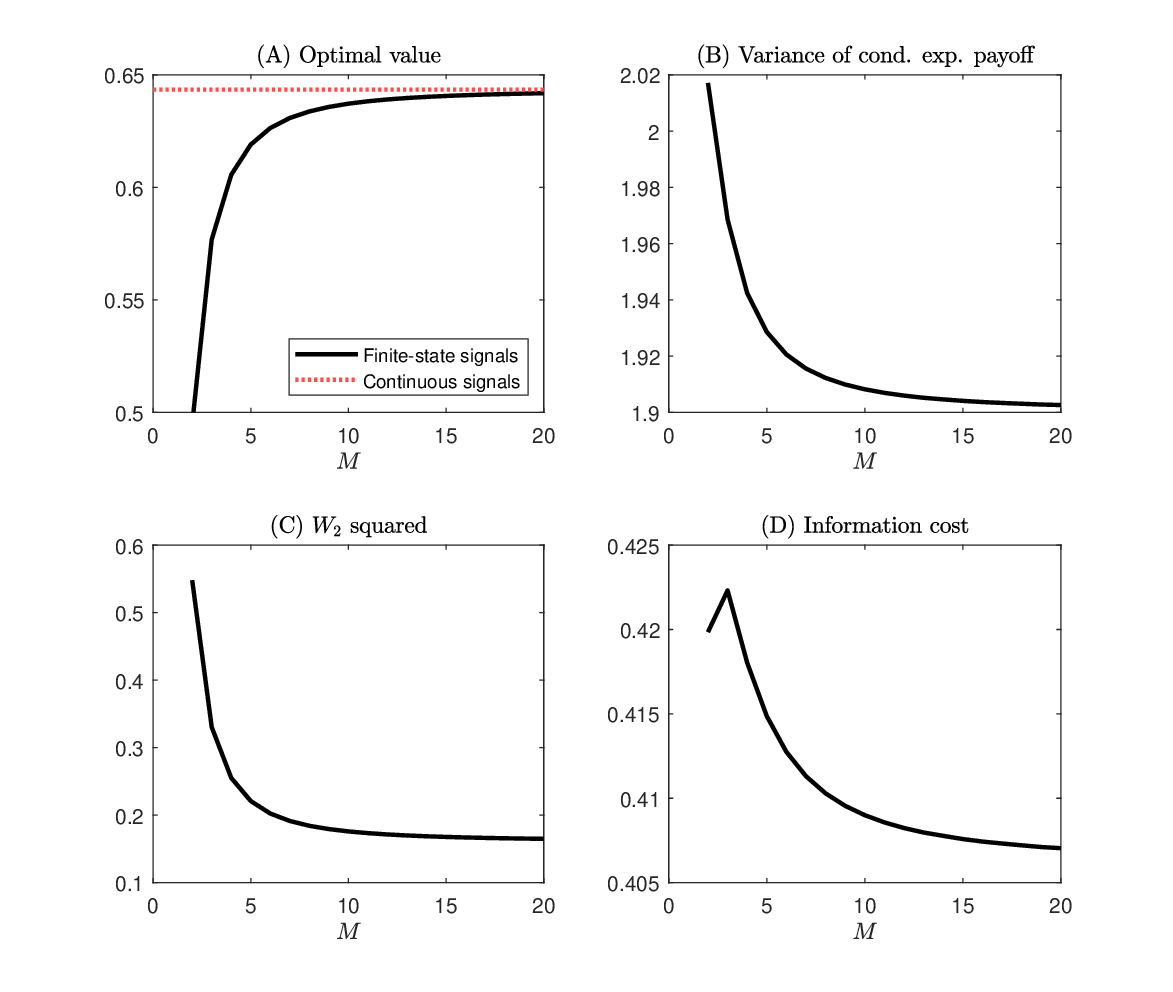}
\caption{Convergence of values (asymmetric prior $p(v_1)=0.3$)}

    \begin{minipage}{\textwidth}
	\footnotesize
	\vspace{2mm}
This figure presents the convergence of optimal value and its different components for the informed trader, as the number os signal states $M$ increases.   
Parameters: $N=2$, $v_1 = -2, v_2 = 2$, $p(v_1) = 0.3, p(v_2) = 0.7$, $\lambda = 2$, $\sigma_Z=1$, and $T=1$.
    \end{minipage}
    \label{fig:conv_asy}
\end{figure}

\end{document}